\documentclass{article}
\usepackage[utf8]{inputenc}

\usepackage{amssymb}
\usepackage{amsmath}
\usepackage{amsthm}
\usepackage{amscd}
\usepackage{algorithm}
\usepackage{algorithmic}
\usepackage{color}
\usepackage{epstopdf}
\usepackage{lineno,hyperref}
\usepackage{lineno,hyperref}
\modulolinenumbers[5]

\newtheorem{theorem}{Theorem}[section]
\newtheorem{lemma}[theorem]{Lemma}
\newtheorem{proposition}[theorem]{Proposition}

\newtheorem{definition}[theorem]{Definition}
\newtheorem{problem}[theorem]{Problem}

\title{Metric Based Quadrilateral Mesh Generation}

\author{Wei Chen\thanks{Dalian University of Technology, Dalian, China. Email: wei.chen@mail.dlut.edu.cn},
Xiaopeng Zheng\thanks{Dalian University of Technology, Dalian, China. Email: zhengxp@dlut.edu.cn},
Jingyao Ke\thanks{University of Science and Technology of China, Hefei, China. Email: keyushu@mail.ustc.edu.cn},
Na Lei\thanks{Dalian University of Technology, Dalian, China. Email: nalei@dlut.edu.cn (corresponding author)},\\
Zhongxuan Luo\thanks{Dalian University of Technology, Dalian, China. Email: zxluo@dlut.edu.cn},
Xianfeng Gu\thanks{Stony Brook University, New York, US. Email: gu@cs.stonybrook.edu}}

\date{ }

\usepackage{graphicx}

\begin{document}

\maketitle
\begin{abstract}
This work proposes a novel metric based algorithm for quadrilateral mesh generating. Each quad-mesh induces a Riemannian metric satisfying special conditions: the metric is a flat metric with cone signualrites conformal to the original metric, the total curvature satisfies the Gauss-Bonnet condition, the holonomy group is a subgroup of the rotation group $\{e^{ik\pi/2}\}$, furthermore there is cross field obtained by parallel translation which is aligned with the boundaries, and its streamlines are finite geodesics. Inversely, such kind of metric induces a quad-mesh. Based on discrete Ricci flow and conformal structure deformation, one can obtain a metric satisfying all the conditions and obtain the desired quad-mesh.

This method is rigorous, simple and automatic. Our experimental results demonstrate the efficiency and efficacy of the algorithm.
\end{abstract}

\section{Introduction}
\begin{figure}[h!]
\centering
\begin{tabular}{cc}
\includegraphics[width=0.5\textwidth]{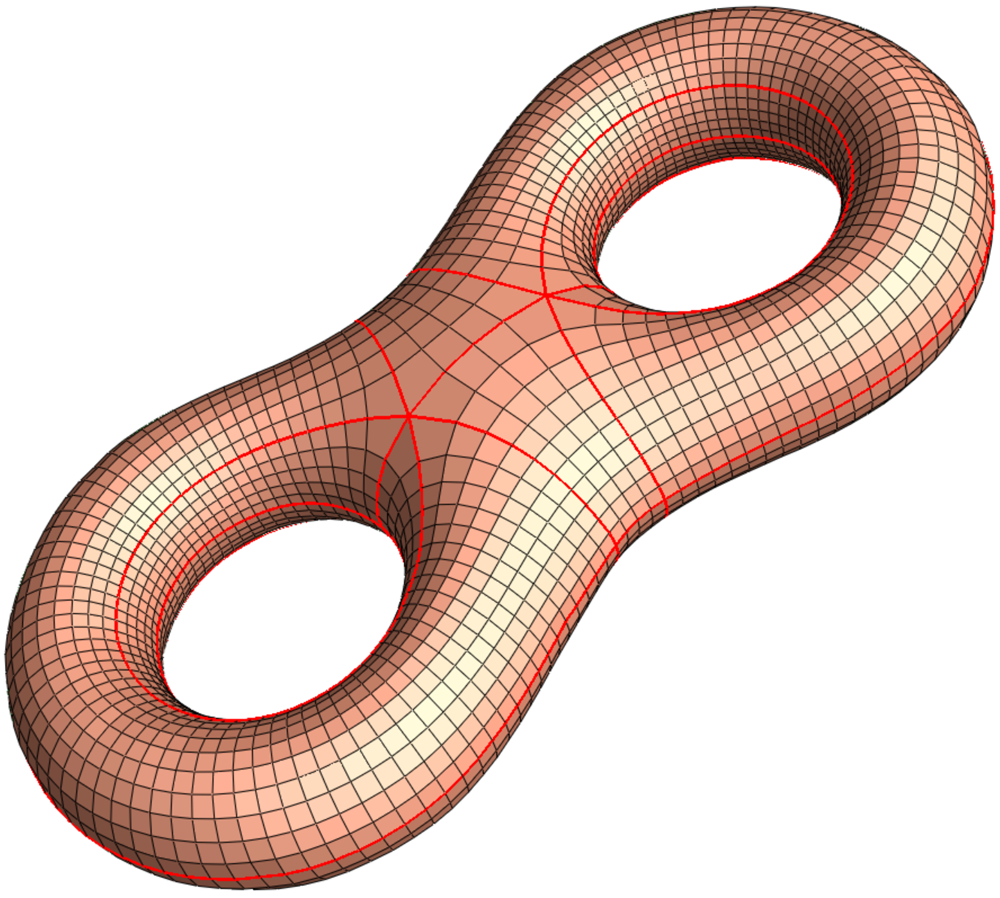}&
\includegraphics[width=0.5\textwidth]{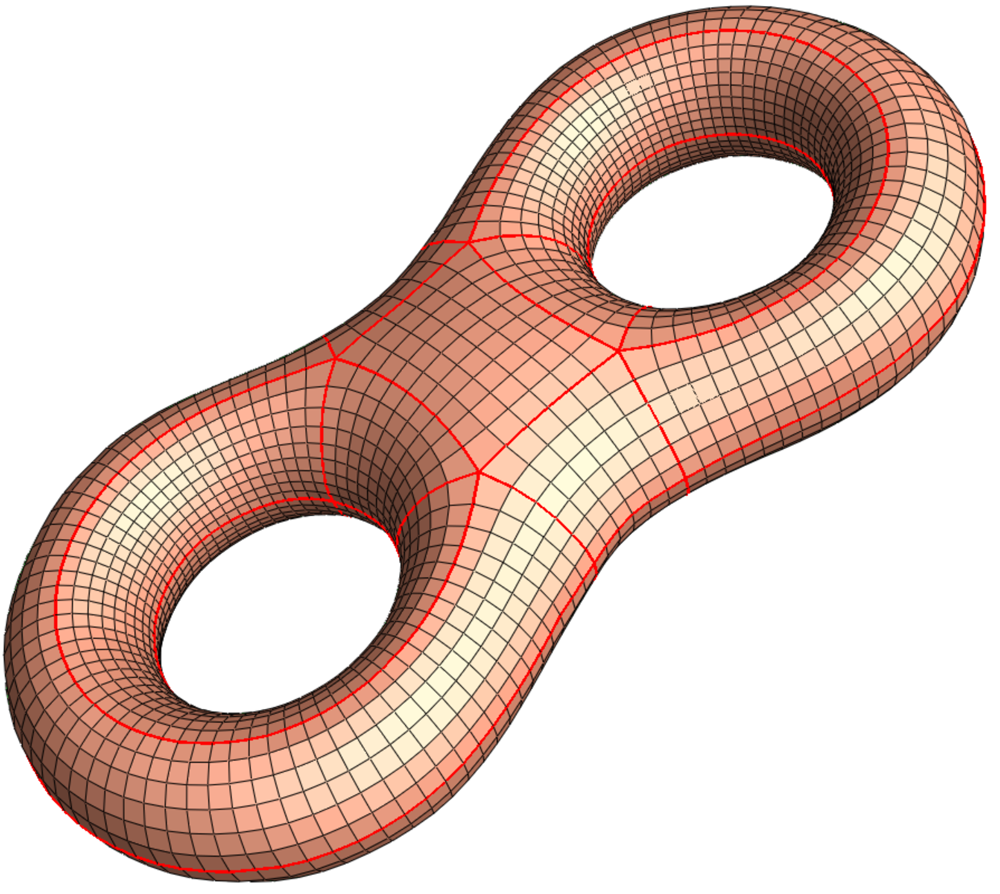}\\
\end{tabular}
\caption{A quad-mesh with 4 and 8 singularities on a genus 2 surface. The red curves are the separatrices, which form the skeletons.}
\label{fig:eight1}
\end{figure}

\paragraph{Quadrilateral Meshes}

With the development of 3D acquisition technologies, triangle meshes become ubiquitous in many engineering fields for its simplicity and flexibility. However, quadrilateral meshes have been widely used in CAD and simulation because they have many merits: 1) quad-mesh has tensor product structure, suitable for Spline fitting purpose. Hence quad-mesh is applied for high-order surface modeling, such as CAD/CAM for Splines and NURBS, and movie industry for subdivision surfaces; 2) quad-mesh better captures the local geometric characteristics, such as principle directions or sharp features, as well as the semantics of the objects, hence it is widely used in animation industry; 3) patches of the skeleton of quad-meshes with a rectangular grid topology match the sampling pattern of textures. Therefore quad-mesh is preferred for texture mapping and compression.

%\begin{itemize}
%\item Quad-mesh has tensor product structure, it is suitable for fitting splines or NURBS. Therefore it is applied for high-order surface modeling, such as CAD/CAM for Splines and NURBS, and the entertainment industry for subdivision surfaces.
%\item Quad-mesh can better capture the local principle curvature directions or sharp features, as well as the semantics of modeled objects, therefore it is widely used in animation industry.
%\item Patches of semi-regular quad meshes with a
%rectangular grid topology, naturally match the sampling
%pattern of textures. Therefore quad-mesh is highly preferred for texturing and compression.
%\end{itemize}

Fig.~\ref{fig:eight1} shows two quad-meshes on a genus two surface. A vertex is called \emph{regular}, if its topological valence is $4$; otherwise, it is \emph{singular}. The left quad-mesh has $4$ singularities, the right quad-mesh has $8$ singularities. The \emph{separatrices} are drawn in red, which are geodesics, roughly speaking the shortest paths on the quad-mesh, connecting singularities. The separatrices divide the surface into rectangular patches, this forms the \emph{skeleton} of the quad-mesh, which is the coarsest level of the quad-mesh. The finer levels of the quad-mesh can be obtained by subdividing the skeleton.

%A path consisting of a sequence of adjacent edges $\{e_0,e_1,\cdots, e_k\}$ is called a \emph{separatrice}, if it starts from and ends at singularities, and each pair of consecutive edges $e_i, e_{i+1}$ share the same regular vertex $v_i$, furthermore $e_i$ and $e_{i+1}$ belong to different faces.

The regularity of the quad-mesh can be described by the number of singularities, and the global behavior of the separatrices. Roughly speaking, quad-meshes can be classified to four categories with the ascending regularity:
%\vspace{-3mm}
\begin{enumerate}
\item Unstructured quad-mesh: a large fraction of its vertices are singularities, the tensor product structure can hardly be found.
\item Valence semi-regular quad-mesh: The number of singularities are few, but the separatrices have complicated global behavior, they may have intersections, form spirals and go through most edges.
\item Semi-regular quad-mesh: The separatrices divide the quad-mesh into several topological rectangles, the interior of each topological rectangle is regular grids.
\item Regular quad-mesh: There are no singularities, all vertices are normal, such as geometry image \cite{Gu:2002:GI:566654.566589}.  A regular quad-mesh has strong topological restriction, it must a topological disk, or an annulus or a torus.
\end{enumerate}
The current work focuses on a special class of quad-meshes, the semi-regular quad-mesh.

%which is between regular and semi-regular categories, which we call \emph{generalized regular quad-mesh}. A generalized regular quad-mesh combines the advantages of both regular and semi-regular quad-meshes, and overcome their disadvantages. Comparing to regular quad-meshes, generalized regular quad-meshes have no topological restrictions; comparing to semi-regular quad-meshes, generalized regular quad-meshes have more regular global structures, which consist of topological cylinders and reduce the number of singularities to the theoretic lower bound.

%Our method is based on the following observations: by subdividing a generalized regular quad-mesh infinitely many times, one obtain two conjugate foliations; the foliations are equivalent to holomorphic quadratic differentials; a holomorphic quadratic differential can be obtained by a graph-valued harmonic mapping. Therefore, one can construct a pants decomposition of the surface, then convert the pants decomposition to a graph, compute the harmonic mapping from the surface to the graph, induce the conjugate foliations from the harmonic mapping, produce the generalized regular quad-mesh from the foliations, as shown in Fig.~\ref{fig:pipeline} and Fig.~\ref{fig:eight_3hole}.

\paragraph{Metric Based Method}

%Surface quadrilateral mesh generation plays a fundamental role in isogeometric analysis. In order to represent geometric data to Spline surfaces, one need to generate the quad-meshes with minimal extraordinary points.

Given a topological surface $S$ with a quad-mesh structure $\mathcal{Q}$, if we treat each quadrilateral face as a unit Euclidean square, then the quad-mesh structure naturally induces a Riemannian metric $\mathbf{g}$, the so-called \emph{quad-mesh metric}. The metric $\mathbf{g}$ induces zero Gaussian curvature everywhere, except at the singularities. If the topological valence of an interior singularity is $k$, then the Gaussian curvature measure at the point is $(1-k/4)\pi$; the Gaussian curvature measure of a boundary singularitiy with valence $k$ is $(1-k/2)\pi$. Therefore, the metric $\mathbf{g}$ is a flat metric with cone singularities.

In practice, it is highly desirable that the quad-faces are uniform squares, this implies that the quad-mesh metric $\mathbf{g}$ is \emph{conformal} to the initial metric, namely these two Riemannian metrics differ by a scalar function. In this work, we study the following problem:

\begin{problem} Given a topological surface $S$, what kind of Riemmannian metric $\mathbf{g}$ is induced by a quad-mesh $\mathcal{Q}$?
\end{problem}

We prove that a metric $\mathbf{g}$ induced by a quad-mesh $\mathcal{Q}$ has many special properties:
\begin{enumerate}
    \item The metric $\mathbf{g}$ is flat except at the singularities. The total curvature measures at the singularities equals to $2\pi$ multiply the Euler characteristic number of the surface. This is the \emph{Gauss-Bonnet condition}.
    \item In each quad-face, we can assign a cross (two orthogonal line segments, parallel to edges), then we get a global smooth \emph{cross field}. This is equivalent to the so-called \emph{holonomy condition}.
     \item If the surface has boundaries, then the cross field is aligned with the boundaries, namely the boundaries are either parallel or orthogonal to the axies of the crosses. This is called \emph{boundary alignment condition}.
    \item If we connect the horizontal and vertical edges of the quad-faces, we get geodesic loops. If we subdivide the quad-mesh infinite many times, we obtain geodesic lamination, each leaf is a closed loop. This is called the \emph{finite geodesic lamination condition}.
\end{enumerate}
Inversely, if we have a metric $\mathbf{g}$ satisfies the above conditions, then the geodesics aligned with the cross field give the quad-mesh $\mathcal{Q}$.

%Furthermore, the surface $S$ with the quad-mesh metric $\mathbf{g}$ can be treated as a Riemann surface. We will show that the metric $\mathbf{g}$ is induced by a special holomorphic (or meromorphic) quartic form, which has local representation $\varphi(z)dz^4$, where $\varphi(z)$ is a holomorphic function.

Equivalently, the surface with the metric $(S,\mathbf{g})$ can be treated as a generalized translation surface.
\begin{definition}[Generalized Translation Surface]
Suppose $P$ is a polygon immersed in the Euclidean plane $\mathbb{R}^2$. The sides of $P$ are identified by the rigid motions of the plane, such that all the rotation angles are $k\pi/2$, $k\in \mathbb{Z}$, then the quotient space is called a generalized translation surface.
\end{definition}

This fact inspires us to develop the metric driven approach for quad-mesh generation. In order to find a high quality quad-mesh $\mathcal{Q}$, we try to find a Riemannian metric $\mathbf{g}$ satisfying the above conditions, then trace the geodesics under $\mathbf{g}$. First, we obtain a flat metric with cone singularities $\mathbf{\tilde{g}}$ using discrete Ricci flow method \cite{Gu_JDG_2018}. The Ricci flow theory guarantees the existence and the uniqueness of the solution. Then we deform the surface with the flat metric $\mathbf{\tilde{g}}$ to satisfies the above conditions. Once the metric is obtained, the geodesics can be calculated using exact geodesic tracing method on polyhedral surfaces \cite{Surazhsky:2005:FEA:1073204.1073228}. Two families of orthogonal geodesics induce the quad-mesh $\mathcal{Q}$.

\paragraph{Contributions}

To the best of our knowledge, this is the first method that generates quad-mesh by designing a Riemannian metric with special properties. The main theorems \ref{thm:main} and \ref{thm:main_inverse} give the equivalence relation between a quad-mesh and its induced Riemannian metric, this gives a novel approach to construct the quad-mesh by finding the metric. Thanks to the solid theory of discrete surface Ricci flow \cite{Gu_JDG_2018}, which guarantees the existence and uniqueness of the solution. The clean and succinct theoretic results make the algorithm pipeline simple and automatic.

The work is organized as follows: section \ref{sec:previous_works} briefly review the most related works; section \ref{sec:theory} introduces the theoretic background, and prove the main theorems; section \ref{sec:algorithm} explains the algorithm in details, and give simple examples to illustrate the key ideas; the experimental results are reported in section \ref{sec:experiments}. Finally, the work concludes in section \ref{sec:conclusion}. 
\section{Previous Works}
\label{sec:previous_works}
The literature of quad-meshing is vast, in the following we only review the most relevant works. For more complete and thorough literature review, we refer readers to \cite{survey:Bommes2013Quad}. There are several approaches for quad-mesh generation.

\paragraph{Triangle Mesh to Quad-Mesh Conversion} The simplest way is to convert a triangular mesh to a quad-mesh directly, then perform Catmull-Clark subdivision. Alternatively, two original adjacent triangles can be fused into one quadrilateral to form a quad-mesh \cite{Gurung2011SQuad,Remacle2012Blossom,Marco2010Practical,Velho20014}. This type method can only produce unstructured quad-meshes, the quad shape is determined by the input triangle mesh.

\paragraph{Patch Based Approach} 

This approach computes the skeleton first, which divides the input surface into several square patches, then subdivides the patches to obtain the quad-mesh. This method can produce semi-regular quad-meshes. The clustering method generates the skeleton by merge neighboring triangle faces into a patch, such as normal-based and center-based methods \cite{Boier2004Parameterization,Carr2006Rectangular}. Poly-cube map \cite{Xia2011Editable,Wang2008User,Lin2008Automatic,He2009A} are adopted to compute the patches. 

\paragraph{Parameterization Based Approach } 
Many quad-meshing algorithms belong to this category.
The spectral surface quadrangulation method \cite{Dong2006Spectral,Huang2008Spectral} produces the skeleton structure from the Morse-Smale complex of an eigenfunction of the Laplacian operator on the input mesh. Discrete harmonic forms \cite{Tong2006Designing}, periodic Global Parameterization \cite{Alliez2006Periodic} and
Branched Coverings method \cite{K2010QuadCover} are all based on parameterization for quad mesh generation.

\paragraph{Voronoi Based Method} The method in \cite{L2010Lp} generates quad-meshes by introducing Lp-Centroidal Voronoi Tessellation (Lp-CVT), which is a generalization of CVT that allows for aligning the axes of the Voronoi cells with a predefined background tensor field. This method can only produce non-structured quad-mesh, there is no global tensor product structure. 

\paragraph{Cross field Based Method} Cross field guided quad-mesh generation was widely studied recently and many approaches have been proposed. Each approach must first choose a way to represent a cross, for example N-RoSy representation\cite{Palacios2007Rotational}, period jump technique\cite{Li2006Representing} and complex value representation\cite{Kowalski2013A}. Then the approaches usually generate a smooth cross field by energy minimization technique. The typical measure of field smoothness is a discrete version of the Dirichlet energy\cite{JFH}. In the end, based on the obtained cross field, these approaches generate the quad meshes by using streamline tracing techniques\cite{RS:RPS:2014} or parameterization method\cite{Bommes2012Quad}.

The cross field guided quad mesh generation method can be very useful and flexible. However it’s not easy to control the position of the singularities and the structures of the quad layout directly. The work in \cite{landau} relates the Ginzberg-Landau theory with the cross field for genus zero surface case.

Comparing to all the existing methods, our method tackles the problem from a complete different angle - the Riemannian metric induced by the quad-mesh. By using discrete Ricci flow, such kind of metric can be obtained with theoretic guarantees.
\section{Theoretic Background}
\label{sec:theory}

\begin{definition}[Quadrilateral Mesh] Suppose $S$ is a topological surface, $\mathcal{Q}$ is a cell partition of $S$, if all cells of $\mathcal{Q}$ are topological quadrilaterals, then we say $(S,\mathcal{Q})$ is a quadrilateral mesh.
\end{definition}

\begin{definition}[face path] A face path in $(S,\mathcal{Q})$ is a sequence $\gamma=(\sigma_0,\sigma_1,\dots,\sigma_n)$ such that $\sigma_0,\sigma_1,\dots,\sigma_n$ are faces and two consecutive faces $\sigma_i$ and $\sigma_{i+1}$ are neighbors in $\mathcal{Q}$ for all $0\le i \le n$.
\end{definition}
The \emph{inverse path} of $\gamma$ is denoted by $\gamma^{-1}=(\sigma_n,\sigma_{n-1},\dots,\sigma_0)$. We write $\gamma\eta$ for the concatenation of $\gamma$ with some face path $\eta=(\sigma_n,\dots,\sigma_m)$. The face path $\gamma$ is \emph{closed} if $\sigma_0=\sigma_n$.

Each face path $\gamma = (\sigma_0,\dots, \sigma_n)$ in $K$ induces a piecewise
linear path $\bar{\gamma}$ in the geometric realization of $(S,\mathcal{Q})$: Join the barycenter of each face $\sigma_i$ by linear paths to the barycenters of the common edges $\sigma_i\cap \sigma_{i-1}$ and $\sigma_i\cap \sigma_{i+1}$ of the neighboring faces $\sigma_{i-1}$ and $\sigma_{i+1}$, respectively. The face path $\gamma$ is closed if and only if the induced piecewise linear path $\bar{\gamma}$ is closed. Often we identify $\gamma$ with $\bar{\gamma}$. Moreover, we write $[\gamma]$ for the homotopy class of $\gamma$ with endpoints fixed.

%\begin{figure}[h!]
%\begin{tabular}{cc}
%\includegraphics[width=0.45\textwidth]{figures/Selection_545.png}&
%\includegraphics[width=0.45\textwidth]{figures/Selection_546.png}}
%\end{tabular}
%\caption{Singularities on a quad-mesh and a hex-mesh. \label{fig:singularities}}
%\end{figure}

On a quad-mesh, the \emph{topological valence} of a vertex is the number of faces adjacent to the vertex.

\begin{definition}[Singularity] Suppose $(S,\mathcal{Q})$ is a quadrilateral mesh. If the topological valence of an interior vertex is $4$, then we call it a \emph{normal vertex}, otherwise a \emph{singularity}; if the topological valence of a boundary vertex is $2$, then we call it a \emph{normal boundary vertex}, otherwise a \emph{boundary singularity}. The index of a singularity is defined as follows:
\[
    \text{Ind}(v_i) = \left\{
    \begin{array}{lcl}
    4-\text{Val}(v_i) & v_i\not\in \partial (S,\mathcal{Q})\\
    2-\text{Val}(v_i) & v_i\in \partial (S,\mathcal{Q})\\
    \end{array}
    \right.
\]
where $\text{Ind}(v_i)$ and $\text{Val}(v_i)$ are the index and the topological valence of $v_i$.
\end{definition}

\subsection{Topological Structure}

%\begin{figure}[h!]
%\begin{center}
%\begin{tabular}{cc}
%\includegraphics[width=0.5\textwidth]{./figures/monodromy.pdf}
%\end{tabular}
%\caption{Parallel transportation along a facet loop. \label{fig:parallel_transport}}
%\end{center}
%\end{figure}
\paragraph{Parallel Transport and Holonomy}

%\emph{Should we add the definition of parallel transport here?}

\begin{definition}[Quadrilateral Mesh Metric]

Given a quadrilateral mesh $(S,\mathcal{Q})$, each quadrilateral face is assigned with the Euclidean metric to be a canonical unit square. This induces a flat metric with cone singularities, denoted as $\mathbf{g}$ and called as the quadrilateral mesh metric of $(S,\mathcal{Q})$.
\end{definition}
Under the quad-mesh metric, the surface is flat, except at the singularities. Let $\Gamma$ be the set of all singularities, then $(S-\Gamma,\mathbf{g})$ is flat everywhere. Hence the parallel transportation under $\mathbf{g}$ in $(S-\Gamma,\mathbf{g})$ is equivalent to the translation in the Euclidean space.

\begin{definition}[Parallel Transportation] Given a quadrilateral mesh $(S,\mathcal{Q})$ with the quad-mesh metric $\mathbf{g}$, $\gamma=(\sigma_0,\sigma_1,\dots,\sigma_{n-1}, \sigma_n)$ is a face path, suppose $\mathbf{v}$ is a tangent vector in $\sigma_0$, at the $i$-th step, $i=1,2,\dots,n$, both $\sigma_{i-1}$ and $\sigma_i$ are isometrically embedded on the Euclidean plane sharing a common edge, then the tangent vector is translated from $\sigma_{i-1}$ to $\sigma_i$. Eventually the tangent vector reaches the $\sigma_n$. The result vector is defined as the parallel transportation of $\mathbf{v}$ along the face path $\gamma$.
\end{definition}

\begin{definition}[Holonomy] Suppose $(S,\mathcal{Q})$ is a quad-mesh with singularity set $\Gamma$ and the quad-mesh metric $\mathbf{g}$. Let $\gamma=(\sigma_0,\sigma_1,\dots,\sigma_{n-1})$ be a face loop. Suppose one choose an orthonormal frame $\{\mathbf{e}_1, \mathbf{e}_2\}$ in $\sigma_0$, where $\mathbf{e}_k$'s are parallel to the edges of $\sigma_0$, and parallel transport the frame along $\gamma$. When the transportation returns to $\sigma_0$ again, the frame becomes $\{\mathbf{\tilde{e}}_1,\mathbf{\tilde{e}}_2\}$. The rotation from the initial frame to the final frame is called the \emph{holonomy} of $\gamma$, and denoted as $\langle\gamma\rangle$.
\end{definition}

Because $(S-\Gamma,\mathbf{g})$ is flat everywhere, if $\gamma_1$ and $\gamma_2$ are homotopic to each other in $S-\Gamma$, then their holonomies are equal, $\langle \gamma_1 \rangle = \langle \gamma_2 \rangle$. The planar rotation group is denoted as $\mathcal{R}=\{e^{k\pi/2},k=0,1,2,3\}$. This induces a homomorphism from the fundamental group of $S-\Gamma$ to the rotation group, $\varphi:[\gamma]\to \langle \gamma \rangle$
\begin{equation}
 \varphi: \pi_1(S-\Gamma,\sigma_0)\to \mathcal{R},
\end{equation}
where $\sigma_0$ is a fixed face. The mapping $\varphi$ is called the \emph{holonomy homomorphism} of the quad-mesh. The image of the holonomy homomorphism
\[
    \Pi(\mathcal{Q},\sigma_0):=\varphi(\pi_1(S-\Gamma,\sigma_0))
\]
is called the \emph{holonomy group} of the quad-mesh. For any face loop $\gamma$, its holonomy $\langle\gamma\rangle$ is a rotation with angle $k\pi/2$, where $k$ is an integer. Therefore, the order of the holonomy group $\Pi(\mathcal{Q},\sigma_0)$ is at most $4$.

\subsection{Riemannian Metric Structure}

%\begin{definition}[Riemann Surface]
%Suppose $S$ is a topological surface, $\mathcal{A}$ is an atlas of $S$ with complex local coordinates. If all the transition functions are biholomorphic, then $\mathcal{A}$ is called a conformal atlas. The surface with a conformal atlas is called a Riemann surface.
%\end{definition}

The quad-metric has many special properties, which are summarized as the following theorem.

\begin{theorem}[Quad-mesh metric]
\label{thm:quad_mesh_metric}
If a Riemannian metric $\mathbf{g}$ with cone singularities is induced by a quad-mesh $(S,\mathcal{Q})$, then it has the following properties:
\begin{enumerate}
    \item The metric $\mathbf{g}$ is flat except at the singularities. The total curvature measures at the singularities equals to $2\pi$ multiply the Euler characteristic number of the surface. This is the \emph{Gauss-Bonnet condition}.
    \item In each quad-face, we can assign a cross (two orthogonal line segments, parallel to edges), then we get a global smooth \emph{cross field}. Namely, the holonomy group is a subgroup of $\mathcal{R}$. This is equivalent to the  \emph{holonomy condition}.
     \item If the surface has boundaries, then the cross field is aligned with the boundaries, namely the boundaries are either parallel or orthogonal to the axes of the crosses. This is called \emph{boundary alignment condition}.
    \item By connecting the horizontal or vertical edges of the quad-faces, geodesic loops can be obtained. If the quad-mesh is subdivided infinite many times, a geodesic lamination is obtained, whose leaves are closed loops. This is called the \emph{finite geodesic lamination condition}.
\end{enumerate}
\label{thm:main}
\end{theorem}

\begin{proof}
\noindent{\emph{Gauss-Bonnet condition}} Each normal vertex has $0$ curvature. Each singular vertex has curvature measure $Ind(v_i) \frac{\pi}{2}$, the total Gaussian curvature satisfies the Gauss-Bonnet theorem:
\[
    \sum_{v_i} Ind(v_i) \frac{\pi}{2} = 2\pi \chi(S).
\]

\noindent{\emph{Holonomy Condition}} The holonomy group of the quad-mesh is a subgroup of $\mathcal{R}$. A cross is invariant under the $\mathcal{R}$ action. We can put a cross at the base face $\sigma_0$, whose two axes are aligned with the edges of the square, and parallel transport to all the faces. This gives a global smooth cross field.

\noindent{\emph{Boundary Alignment}} All the boundaries of the quad-mesh consists of the edges of square faces, therefore the cross axes are parallel or orthogonal to the boundaries.

\noindent{\emph{Finite Geodesic lamination}} We start from the center of a face, issue a geodesic parallel with the edges of the face. The geodesic won't enter the same face more than two times. The number of faces is finite, therefore, the geodesic is of finite length.
This holds for all the geodesics constructed this way.
\end{proof}

\begin{theorem}[Inverse Quad-mesh metric Theorem]
\label{thm:quad_mesh_metric_inverse}
Given a topologicla surface $S$ and a flat metric $\mathbf{g}$ with cone singularities $\gamma$, $\mathbf{g}$ has the following properties:
\begin{enumerate}
    \item The metric $\mathbf{g}$ is flat except at the singularities. The total curvature measures at the singularities equals to $2\pi$ multiply the Euler characteristic number of the surface. This is the \emph{Gauss-Bonnet condition}.
    \item The holonomy group is a subgroup of $\mathcal{R}$. This is the  \emph{holonomy condition}.
     \item There is a cross field obtained by parallel transporting a cross defined at one normal point of $S$, such that the cross field is aligned with the boundaries. This is the \emph{boundary alignment condition}.
    \item The stream lines parallel to the cross field are finite geodesic loops. This is the \emph{finite geodesic lamination condition}.
\end{enumerate}
Then a quadrilateral mesh can be constructed on $S$, such that the quad-mesh metric is $\mathbf{g}$.
\label{thm:main_inverse}
\end{theorem}
\begin{proof}
if we have a metric $\mathbf{g}$ satisfies the above conditions, then the geodesics aligned with the cross field give the quad-mesh $\mathcal{Q}$. The geodesics through the singularities are the separatrices.
\end{proof}

\begin{figure}[h!]
\centering
\begin{tabular}{cc}
\includegraphics[width=0.85\textwidth]{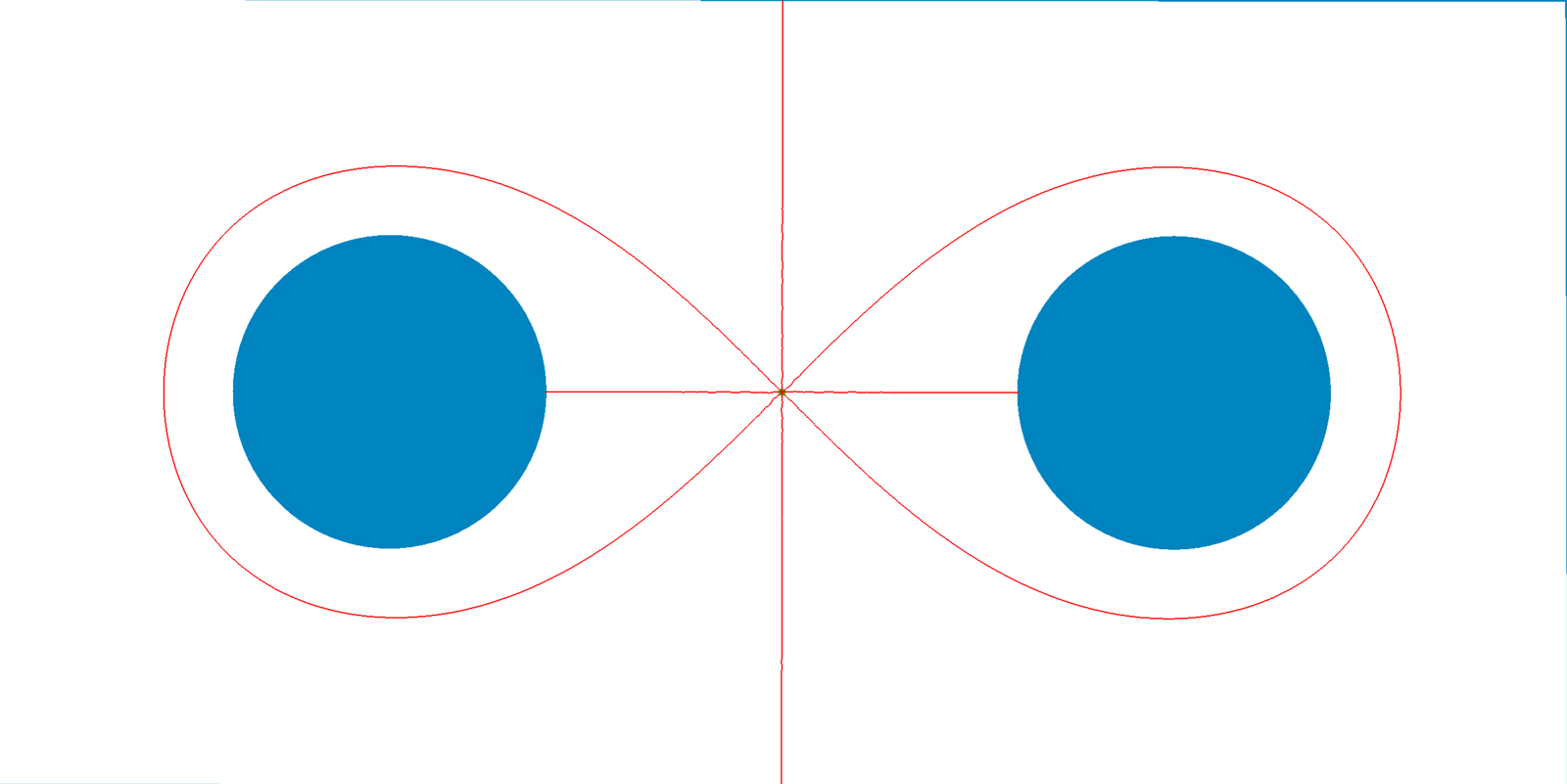}
\end{tabular}
\caption{Gauss-Bonnet condition: a planar domain with two inner boundaries. The center singularity is with index $-4$, everywhere else the curvature is $0$. The geodesic curvature along the boundaries are also $0$. The geodesics through the singularity are drawn as red curves.}
\label{fig:1saddle}
\end{figure}

As shown in Fig.~\ref{fig:1saddle}, given a planar rectangle with two circular holes, a special flat metric is computed with a single singularity, whose index is $-4$. The curvautre is $0$ every where else, including the boundaries. Therefore, the total curvature is $-2\pi$, the Euler characteristic number is $-1$, the Gauss-Bonnet formula holds. The red curves are geodesics through the singularity, they are perpendicular to the boundaries, or form geodesic loops. They are either parallel or orthogonal to each other.

\begin{figure}[h!]
\centering
\begin{tabular}{cc}
\includegraphics[width=0.85\textwidth]{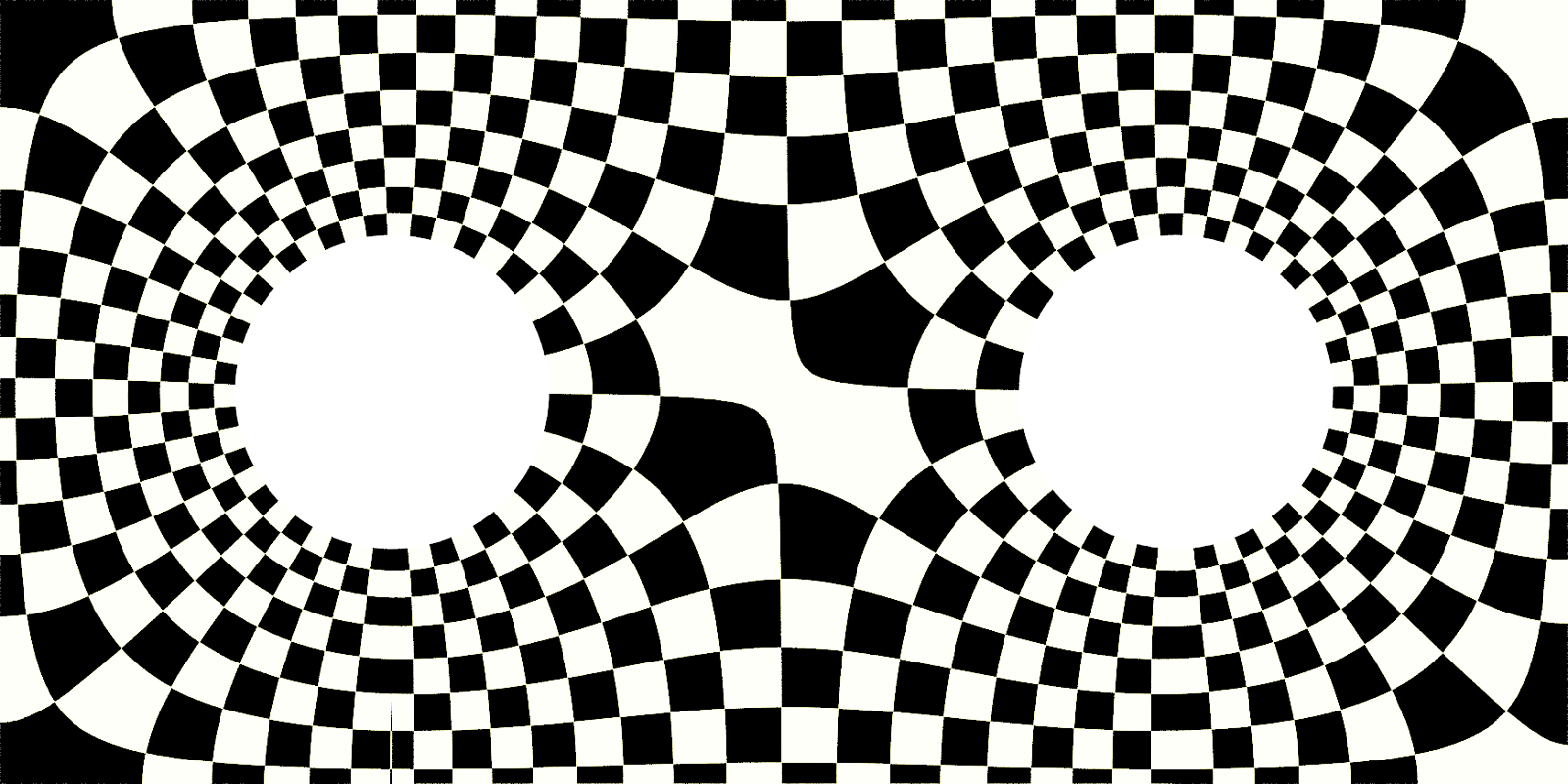}\\
\end{tabular}
\caption{Holonomy condition: the holonomy group of the quad-mesh is trivial.}
\label{fig:1saddle_qud_mesh}
\end{figure}

Fig.~\ref{fig:1saddle_qud_mesh} shows the holonomy condition. The quad-mesh is depicted by checker-board texture mapping. Each checker represents a quadrilateral face. The parallel transportation along two inner boundaries induces trivial holonomy. Similarly, the holonomy of the loop surrounding the singularity is also trivial.

\begin{figure}[h!]
\centering
\begin{tabular}{cc}
\includegraphics[width=0.85\textwidth]{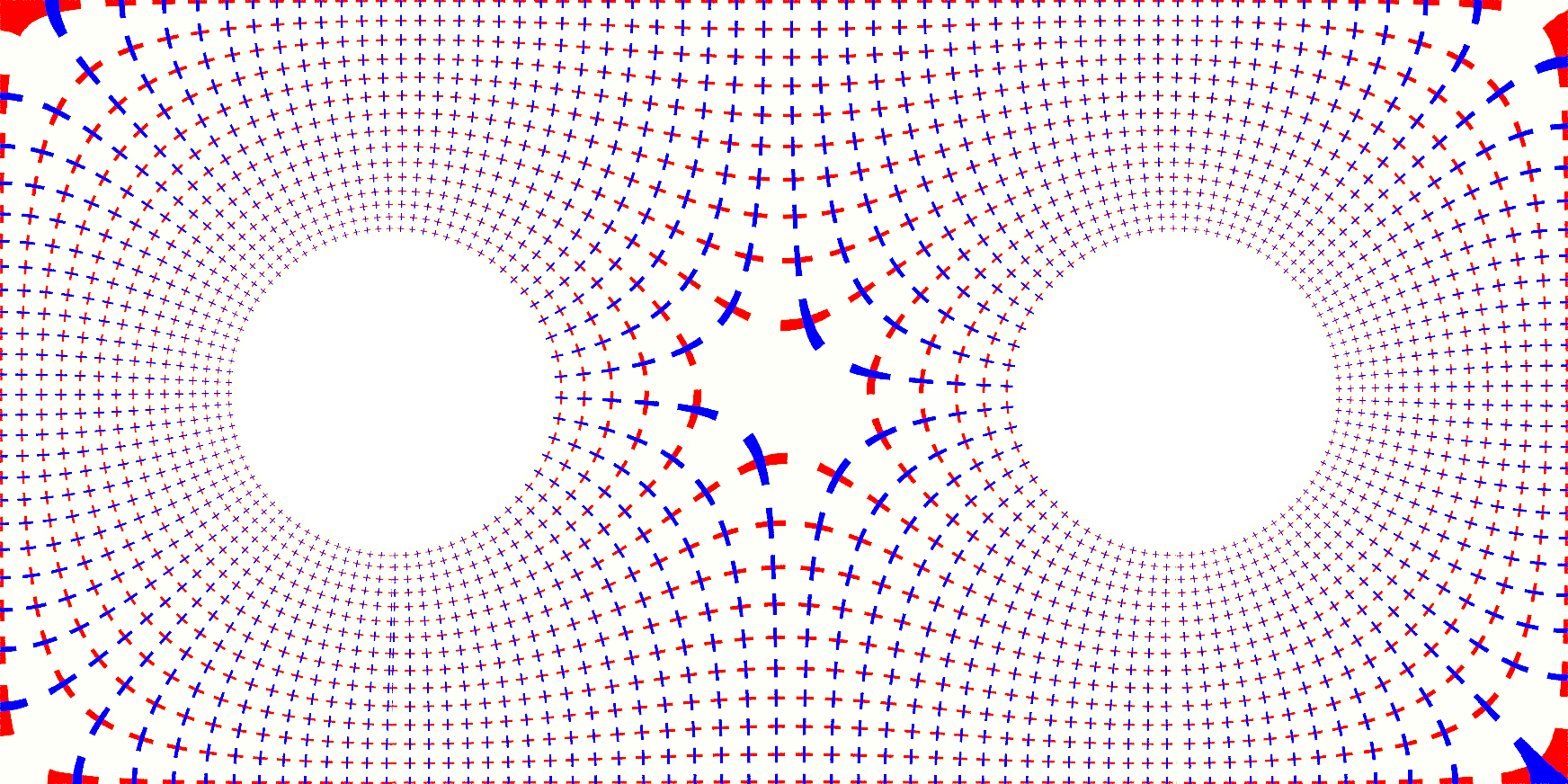}
\end{tabular}
\caption{Boundary alignment condition: the cross field is aligned with all the boundaries.}
\label{fig:1saddle_cross_field}
\end{figure}
Fig.~\ref{fig:1saddle_cross_field} shows the boundary alignment condition. We put a cross in each face, whose axis is aligned with the edges, then we obtain a smooth cross field. The cross field is aligned with all the boundaries.

\begin{figure}[h!]
\centering
\begin{tabular}{cc}
\includegraphics[width=0.85\textwidth]{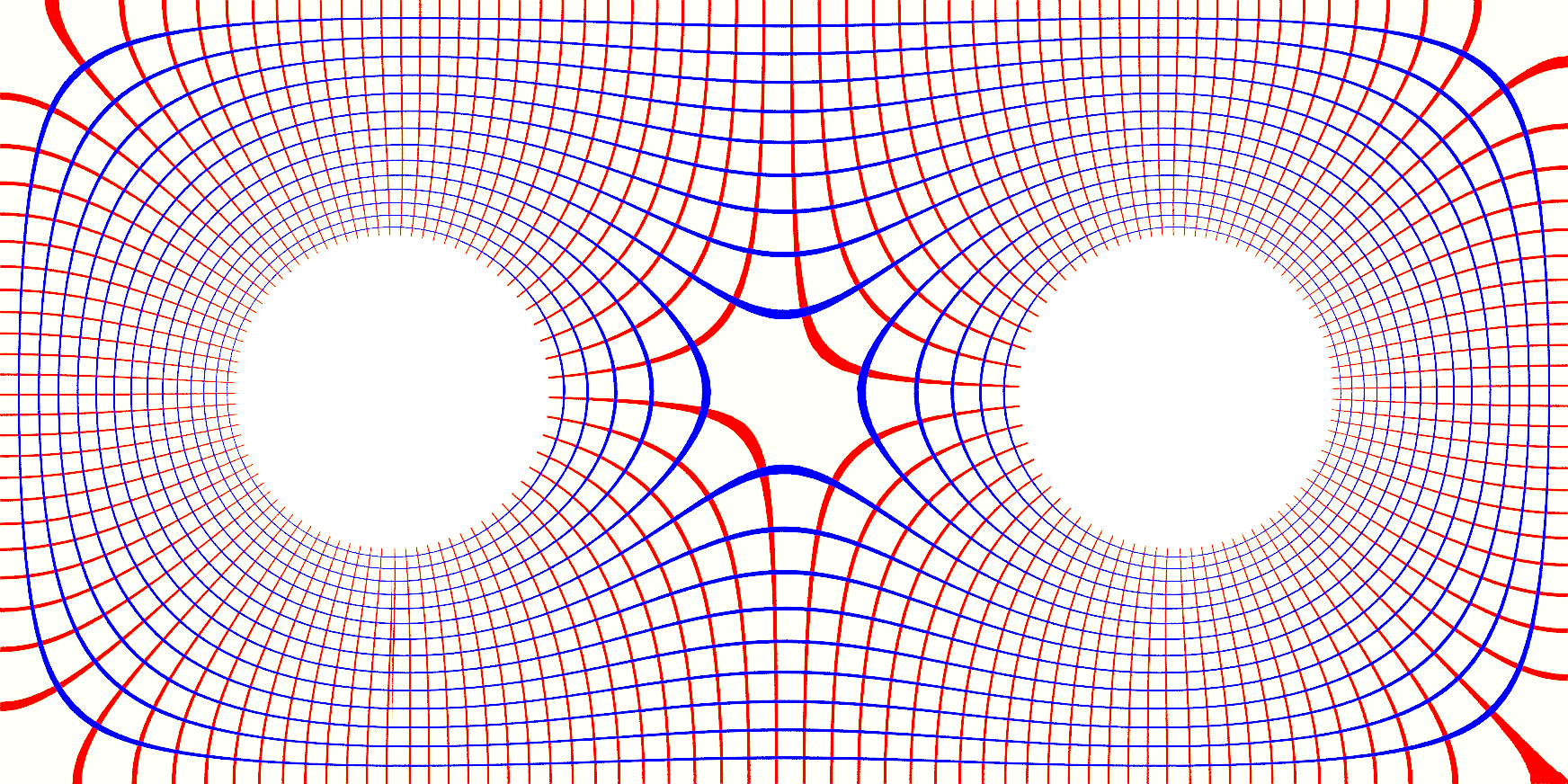}
\end{tabular}
\caption{Finite geodesic condition: all the geodesics aligned with the cross field are finite.}
\label{fig:1saddle_geodesic}
\end{figure}

Fig.~\ref{fig:1saddle_geodesic} shows the finite geodesic condition. All the geodesics parallel to the edges of the faces either terminate at the boundaries or the singularity, or form loops.

\begin{figure}[h!]
\centering
\begin{tabular}{cc}
\includegraphics[width=0.98\textwidth]{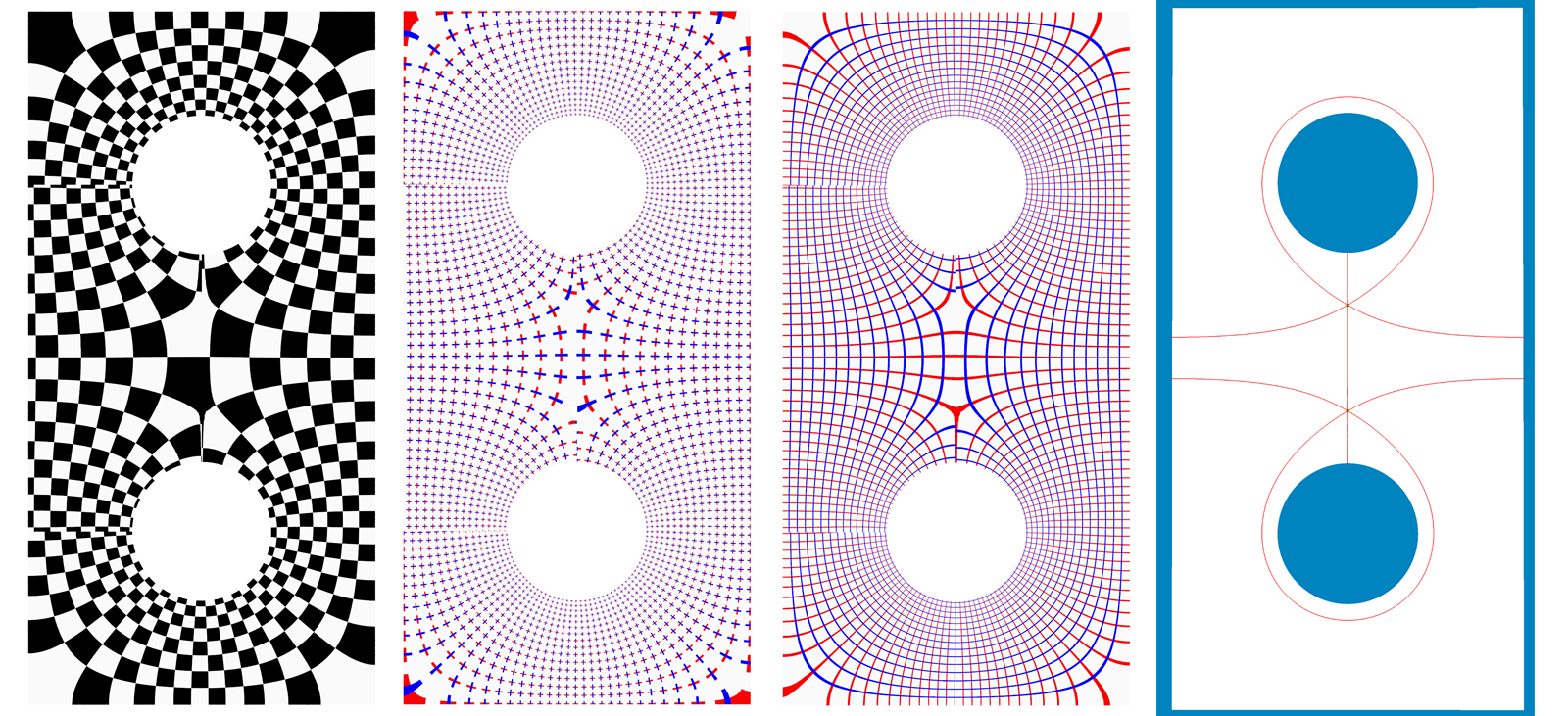}
\end{tabular}
\caption{Two singularity configuration, corresponding to a holomorphic quadratic form.}
\label{fig:2saddle}
\end{figure}

Fig.~\ref{fig:2saddle} illustrates the same surface with $2$ valence-6 singularities, each has $-\pi$ Gaussian curvature measure. From left to right, the quad-mesh, the cross field, the geodesics, the singularities and the geodesics through them and perpendicular to the boundaries. The flat metric with the cone singularities satisfies all the 4 conditions.

\begin{figure}[h!]
\centering
\begin{tabular}{cc}
\includegraphics[width=0.98\textwidth]{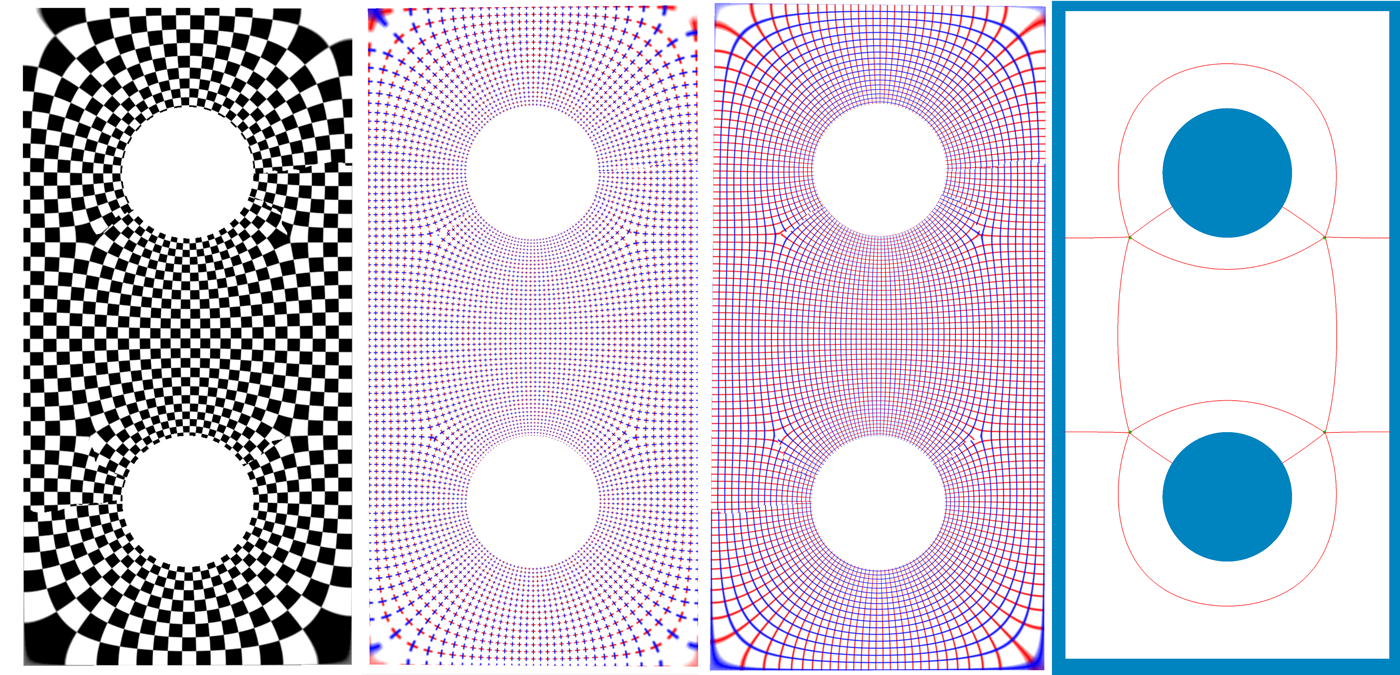}
\end{tabular}
\caption{Four singularity configuration, corresponding to a holomorphic quartic form.}
\label{fig:4saddle}
\end{figure}

Fig.~\ref{fig:4saddle} shows the same surface with $4$ valence-5 singularities, each has $-\frac{\pi}{2}$ Gaussian curvature measure. The flat metric with cone singularities satisfies all the 4 conditions.

\if 0
\begin{proposition}
Suppose $S$ is a genus zero surface with boundaries, $\mathbf{g}$ is a flat metric with cone singularites $\Gamma$. Each singularity has Gaussian curvature measure $\frac{k}{2}\pi$, $k\in\mathbb{Z}$. The total curvature of each boundary component is $2k\pi$, $k\in \mathbb{Z}$. Then the holonomy group induced by $\mathbf{g}$ is a subgroup of $\mathcal{R}$.
\end{proposition}
This means for genus zero surfaces, the Gauss-Bonnet condition and boundary curvature condition imply the holonomy condition.

\begin{proof}
Assume the boundary components of the surface $S$ are
\[
    \partial S = \gamma_0 - \gamma_1 \cdots - \gamma_n,
\]
singularity set is $\Gamma=\{v_1,v_2,\cdots,v_m\}$, each $v_i$ has Guassian curvature $K_i=k_i\frac{\pi}{2}$.

Because the surface is of genus zero, the surface has no handles. Fix a base face $\sigma_0$, the fundamental group $\pi_1(S-\Gamma,\sigma_0)$ has generators $\gamma_1,\gamma_2,\dots,\gamma_n$, and $\tau_1,\tau_2,\dots,\tau_m$, where $\tau_k$ is a loop surrounding the singularity $v_k$. The holonomy of each generator, $\langle \gamma_i \rangle$ or $\langle \tau_j \rangle$, is in $\mathcal{R}$.

\end{proof}
\fi

\begin{figure}[h!]
\centering
\begin{tabular}{cc}
\includegraphics[width=0.98\textwidth]{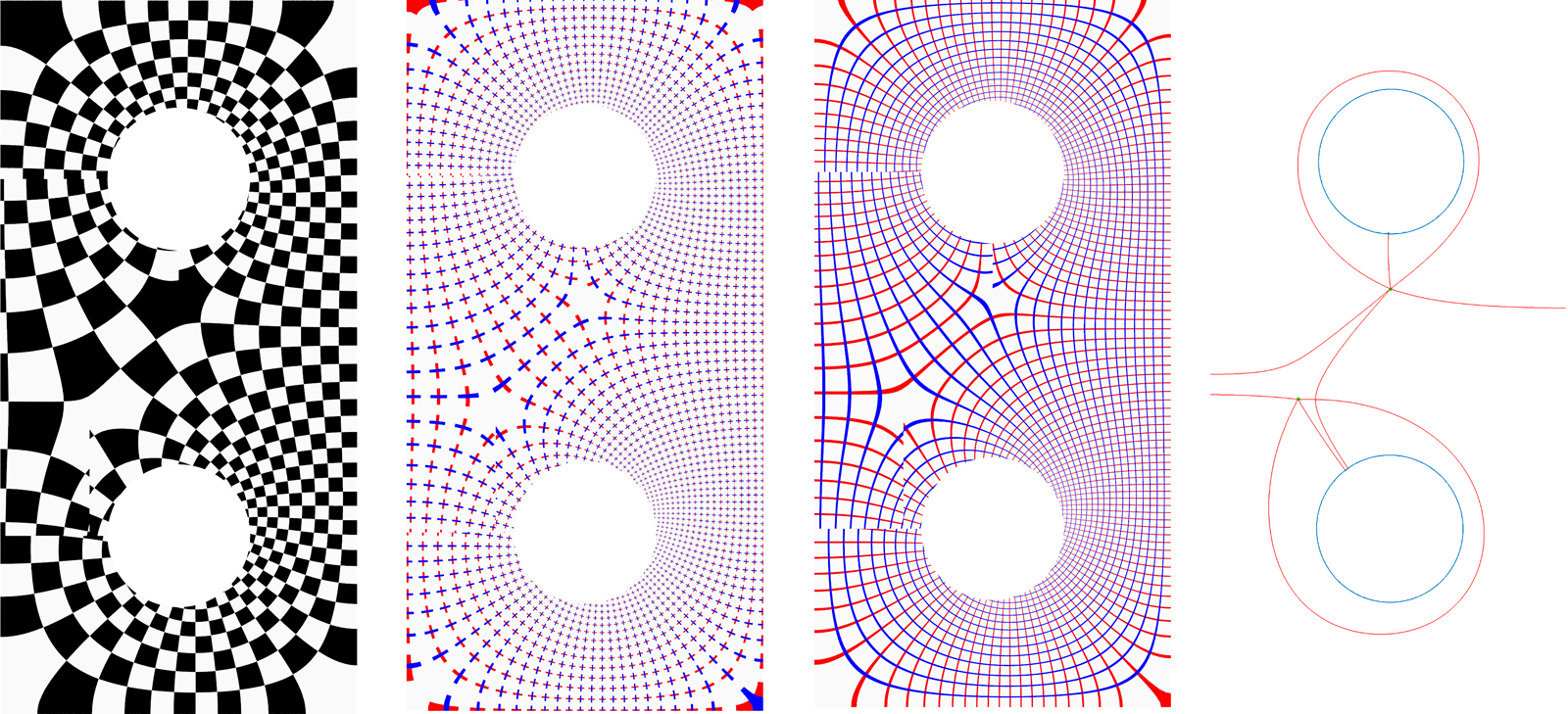}
\end{tabular}
\caption{Two singularity configuration, violating the boundary alignment condition.}
\label{fig:2saddle_twisted}
\end{figure}

Fig.~\ref{fig:2saddle_twisted} shows the same surface with different positions of singularities, the flat metric $\mathbf{g}$ satisfies the Gauss-Bonnet condition. The global smooth cross field in the 2nd frame shows the metric satisfies the holonomy condition. But the cross fields are not aligned with the inner boundaries, hence the geodesics are not parallel or orthogonal to the inner boundaries as shown in the 3rd frame.

\begin{figure}[h!]
\centering
\begin{tabular}{cc}
\includegraphics[height=0.45\textwidth]{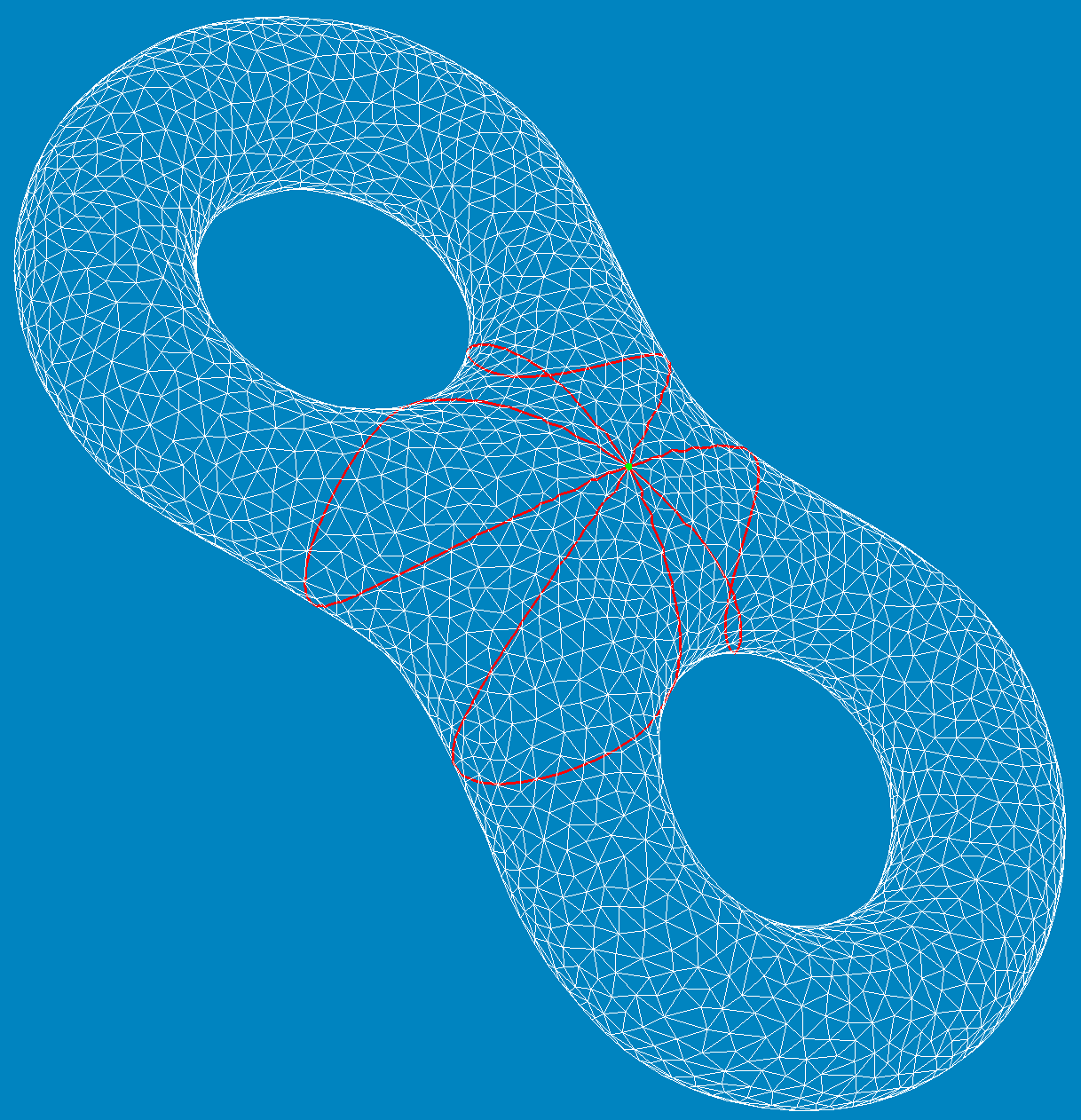}&
\includegraphics[height=0.45\textwidth]{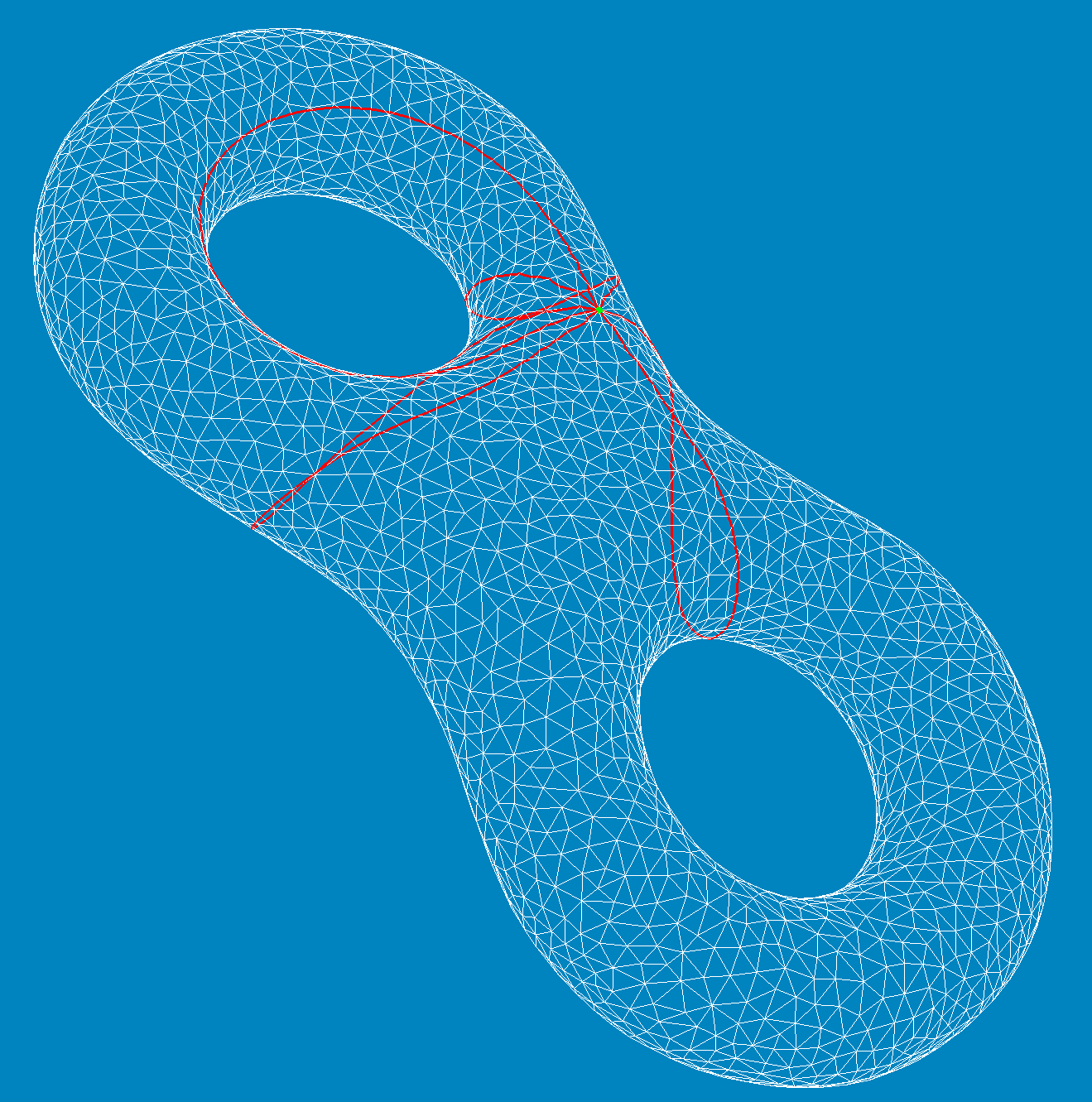}
\end{tabular}
\caption{A flat metric with a cone singularity, whose curvature measure is $-4\pi$. The metric violates the holonomy condition.}
\label{fig:eight_holonomy}
\end{figure}

\begin{lemma} Suppose $S$ is a genus zero surface, $\mathbf{g}$ is a flat metric with internal cone singularities $\Gamma_1=\{p_1,p_2,\cdots,p_n\}$; along the boundary components there are boundary singularities
$\Gamma_2=\{q_1,q_2,\cdots,q_m\}$,
where the discrete Gaussian curvature measure of $p_i$ is $\frac{k_i}{2}\pi$, where $k_i\in \mathbb{Z}$ is an integer, the curvature measure of $q_j$ is $\frac{l_j}{2}\pi$, $l_j\in \mathbb{Z}$.
Furthermore $\mathbf{g}$ satisfies the Gauss-Bonnet condition,
\[
    \sum_{i=1}^{n} k_i + \sum_{j=1}^{m} l_j = 4\chi(S-\Gamma),
\]
then $\mathbf{g}$ also satisfies the holonomy condition.
\end{lemma}
\begin{proof}
Suppose the boundary components of the surface $S$ are
\[
    \partial S = \gamma_0 - \gamma_1 \cdots \gamma_t,
\]
and the singularity set is $\Gamma=\{p_1,p_2,\cdots,p_n\}$. Let $\alpha_j$ is the loop around $p_j$, then the fundamental group of the surface is given by
\[
    \pi_1(S-\Gamma_1) = \langle \gamma_1, \gamma_2, \cdots, \gamma_t, \alpha_1,\cdots, \alpha_n\rangle.
\]
The holonomies of the generators
\[
    \langle \alpha_i \rangle = \frac{k_i}{2}\pi,
\]
the holonomy of each boundary component $\gamma_j$ equals to the total Gaussian curvature measures of all corner singularities along $\gamma_j$. Hence, the metric satisifies the holonomy condition.
\end{proof}

But if the surface is of high genus, then a flat metric $\mathbf{g}$ satisfying Gauss-Bonnet may not satisfy the holonomy condition. Fig.~\ref{fig:eight_holonomy} shows a metric satisfies the Gauss-Bonnet condition, but violates the holonomy condition. We use discrte Ricci flow method to compute a flat metric $\mathbf{g}$ on a genus two surface, such that the unique singularity $p$ is with $-4\pi$ Gaussian curvature. According to the discrete uniformization theorem proved in \cite{Gu_JDG_2018}, such kind of metric exists and is unique upto scaling. We calculate several geodesic loops through the singularity $p$. Suppose $\gamma$ is a geodesic loop, $\gamma(0)=\gamma(1)=p$. If the metric $\mathbf{g}$ satisfies the holonomy condition, then the angle between two tangent vectors $\gamma'(0)$ and $\gamma'(1)$ is $\frac{k}{2}\pi$ under $\mathbf{g}$, where $k$ is an integer. We measure such angles of several geodesic loops through $p$, most of them are not $\frac{k}{2}\pi$. Hence the metric $\mathbf{g}$ doesn't satisfy the holonomy condition.

\section{Algorithm}
\label{sec:algorithm}

The metric based quad-mesh generation aims at computing a flat cone metric with singularities, satisfying the condition in theorem \ref{thm:quad_mesh_metric}, then find two families of orthogonal geodesics to generate the quadrilateral mesh.

\subsection*{Algorithmic Pipeline}

Suppose the input surface $S$ is discretized as a triangular mesh. The algorithm pipeline is as follows:
\begin{enumerate}
    \item Determine the positions and indices of singularities $\Gamma$;
    \item Compute a flat metric $\mathbf{g}_0$ with cone singularities using discrete surface Ricci flow algorithm;
    \item Compute a cut graph $L$ of the surface, such that $S-L$ is a topological disk. Furthermore, for each singularity $v_i\in \Gamma$, find the shortest path connecting $v_i$ and the boundary of $S-L$, the shortest paths are added to $L$;
    \item Isometrically immerse $(S-L,\mathbf{g}_0)$, the image is a planar immersed polygon $P$. Each pair of dual boundary segments of the polygon  differ by a planar rigid motion.
    \item Conformal structure deformation. Adjust the boundary of $P$, such that each pair of dual boundary segments of $P$ differ by a translation and a rotion in $\mathcal{R}$, and all the boundary segments of $S$ are horizontal or vertical. Use harmonic map to deform the interior of $P$. This induces a new flat metric $\mathbf{g}$, and a cross field $\omega$ satisfying the boundary condition.
    \item Compute geodesics under $\mathbf{g}$ align the cross field $\omega$. The geodesics through the singularities defines the skeleton, further subdivisions of the skeleton gives the quad-mesh.
\end{enumerate}

In the following, we explain each step in details.

\subsection{Singularity Location}
The most crucial step of the algorithm pipeline is to determine the positions and indexes of the singularities. One way is to manually input the positions and indexes using heuristics, then generate the skeleton. In general, the singularity configurations need to be adjusted in order to improve the mesh quality.

An automatic way to determine the singularities to use the poles and zeros of an Abel differential on the surface. The details will be introduced in our later submission \cite{Zheng_2019}.

%is based on the theorem \ref{}. First we choose a holomorphic/meromorphic quartic differential $\omega$, locate its zeros and poles, determine their indexes. The zeros and poles of $\omega$ are the singularities. In practice, we compute the basis of all holomorphic quadratic differentials on the surface using the algorithm in \cite{}. Then we select two holomorphic quadratic differentials $\omega_1$ and $\omega_2$, their product $\omega=\omega_1\omega_2$ is a holomorphic quartic differential. The zeros/poles of $\omega$ is the union of the zeros/poles of $\omega_1$ and $\omega_2$. In this way, we can determine the singularity automatically.

\subsection{Discrete Surface Ricci Flow}
Given a polyhedral surface with a triangulation, the surface has induced Euclidean metric, namely, a triangle mesh $M$.

Each face is a Euclidean triangle $[v_i,v_j,v_k]$ with edge lengths $\{l_i,l_j,l_k\}$. The corner angles and the edge lengths satisfies the cosine law:
\[
    l_i^2 = l_j^2 + l_k^2 - 2 l_jl_k\cos \theta_i.
\]
The discrete vertex Gaussian curvature is defined as angle deficit,
\[
    K(v_i) = \left\{
    \begin{array}{rl}
    2\pi - \sum_{jk}\theta_i^{jk}& v_i\not\in \partial M\\
    \pi - \sum_{jk}\theta_i^{jk}& v_i\not\in \partial M\\

    \end{array}
    \right.
\]
The total Gaussian curvature satisfies the Gauss-Bonnet theorem,
\[
    \sum_{i} K(v_i) = 2\pi \chi(M).
\]

We associate each vertex $v_i$ with a discrete conformal factor $u_i$, then the vertex scaling operator is defined as
\[
    l_{ij} = e^{u_i}\beta_{ij} e^{u_j},
\]
where $l_{ij}$ is the length of the edge $[v_i,v_j]$, $\beta_{ij}$ is the initial edge length. Given target curvature $\bar{K}:V\to\mathbb{R}$, the discrete Ricci flow is defined as
\[
    \frac{d u_i}{dt} = \bar{K}_i - K_i.
\]
During the flow, the triangulation is updated to be Delaunay. Discrete Ricci flow is the gradient flow of the convex energy,
\[
    E(u_1,u_2,\dots,u_n) = \int_{0}^{(u_1,\dots,u_n)} \sum_{i=1}^n (\bar{K}_i - K_i) du_i.
\]
This energy can be optimized using Newton's method.

\subsection{Cut Graph}

Given a triangle mesh $M$ with singularities, the singularity set is denoted as $\Gamma$. First, we compute the cut graph $G$ of the mesh $M$.

Let $\bar{M}$ be the dual mesh of $M$, then we compute a spanning tree $\bar{T}$ of $\bar{M}$. The cut graph $G$ is defined as
\[
    G:=\{e\in M| \bar{e}\not\in \bar{T}\},
\]
then $M-G$ is a topological disk. The for each singularity $p_k\in \Gamma$, find the shortest path $\gamma_k$ from $p_k$ to the cut graph, then
\[
    L = G  \bigcup_{p_i\in \Gamma} \gamma_i
\]
$\tilde{M}=M-L$ is a topological disk.

\subsection{Isometric Immersion}

We can isometrically immerse $\tilde{M}$ with the flat metric $\mathbf{g}$ onto the plane, the immersion is denoted as $\varphi: \tilde{M}\to\mathbb{R}^2$. Then $\varphi$ assigns planar coordinates for each vertex in $\tilde{M}$.

Each triangle face of $\tilde{M}$ corresponds to a face of $M$, each vertex $\tilde{v}_i\in \tilde{M}$ corresponds to a unique vertex $v_j \in M$. This defines a simplicial projection map $\psi:\tilde{M}\to M$, which is a covering map.

Each face $f_i$ consists of three vertices, the face-vertex pair $(f_i,v_j)$ represents the corner in $f_i$ with $v_j$ as the apex. Then the corners of $\tilde{M}$ and those of $M$ have one-to-one correspondence. In the covering mesh $\tilde{M}$,
we define the texture coordinates of a corner $(f_i,v_j)$ as the texture coordinates of the vertex $v_j$. In the base mesh $M$, the texture coordinates of a corner equal to the texture coordinates of its preimage.

\subsection{Conformal Structure Deformation}

Let $e_{ij}\in M$ be an edge, it is adjacent to two faces $f_k=[v_i,v_j,v_k]$  and $f_l=[v_j,v_i,v_l]$, the chart transition mapping $\varphi_{ij}:\mathbb{R}^2\to\mathbb{R}^2$ satisfies the condition
\[
    \varphi_{ij}(\varphi(f_k,v_i)) =\varphi(f_l,v_i),
    \varphi_{ij}(\varphi(f_k,v_j)) =\varphi(f_l,v_j).
\]
The chart transition mappings are identity on edges not in $L$.

Consider the graph $L\subset M$, each vertex $v_i\in L$ has a topological valence in $L$, which is the number of edges in $L$ and adjacent to $v_i$. We call vertices in $L$ with valence $2$ as normal vertices, otherwise as \emph{nodes}. The nodes divide the graph $L$ into \emph{segments}. Each $\gamma_i$ is a segment, the boundary $\partial M$ is divided into segments, denoted as $\{\tau_1,\tau_2,\dots, \tau_k\}$.

In $\tilde{M}$, the preimage of $L$ becomes the boundary $\partial \tilde{M}$. Each $\gamma_i$ has two preimages, denoted as $\tilde{\gamma}_i^+$ and $\tilde{\gamma}_i^-$, each $\tau_j$ has a unique preimage $\tilde{\tau_j}$. $\varphi(\tilde{\gamma}_i^+)$ and
$\varphi(\tilde{\gamma}_i^-)$ differ by a planar rigid motion.

We can modify the coordinates of $\varphi(\tilde{\gamma}_i^+)$, $\varphi(\tilde{\gamma}_i^-)$, and $\varphi(\tilde{\tau}_j)$ by rotation and translation , such that
\begin{itemize}
    \item All $\varphi(\tilde{\tau}_j)$'s are horizontal or vertical;
    \item $\varphi(\tilde{\gamma}_i^+)$ and $\varphi(\tilde{\gamma}_i^-)$ differ by a translation and a rotation in $\mathcal{R}$.
\end{itemize}
After modifying the boundary coordinates, we calculate the coordinates of the interior vertices of $\tilde{M}$ using a harmonic map. The new local coordinates gives a new Riemannian metric $\mathbf{g}$, which satisfies the holonomy condition and the boundary alignment condition.
%Then we define $dz^4$ as the holomorphic quartic differential on $\tilde{M}$. The push-forward by the projection map $p_*(dz^4)$ is the desired holomorphic quadrtic differential $\omega$ on $M$.

\subsection{Tracing Geodesics}

%It is convenient to use the local coordinates to trace a geodesic. On each triangle face, a geodesic is a straight line segment. Given starting point on the boundary of the triangle, and the direction on the plane, we can compute the second intersection point between the line and the boundary of the triangle. When we transit from one face to another adjacent face, we can use the transition mapping to transfer from one local parameters to another. Therefore, we can trace the geodesic face by face, until the geodesic reaches the boundary of the mesh or returns to the starting point.
%In the current project, we are only interested in those geodesics, which are horizontal or vertical lines on each local parameters.

We compute the exact geodesic on polyhedral mesh with the metric $\mathbf{g}$ using the algorithm in \cite{Surazhsky:2005:FEA:1073204.1073228}. First, we compute the geodesics issued from the singularities, which are orthogonal to the boundaries, or connected with other singularities. We call these special geodesics as \emph{critical trajectories}.

The critical trajectories are the separatrices, which partition the surface into Euclidean rectangles under the metric $\mathbf{g}$. These rectangles are the coarsest level of the quad-mesh, or the \emph{skeleton} of the quad-mesh. Then we subdivide the skeleton to form the refined quad-mesh.

%\paragrpah{Algorithm}
%Quadrilateral mesh generation plays a fundamental role in engineering fields. This work introduces a novel method based on surface Ricci flow.

%The proposed algorithm pipeline is as follows: first, the user determines the extraodinary points, including both positions and indices; second, a flat metric is computed using surface Ricci flow, the target curvature is determined by the singularities; third by tracing the geodesics on the flat metric, we can build an orthoganol network on the surface, which induces the desired quad-mesh.

\subsection{One Computational Example}

We use a simple example to illustrates the basic computation pipeline.

\begin{figure}[h!]
\centering
\includegraphics[width=0.85\textwidth]{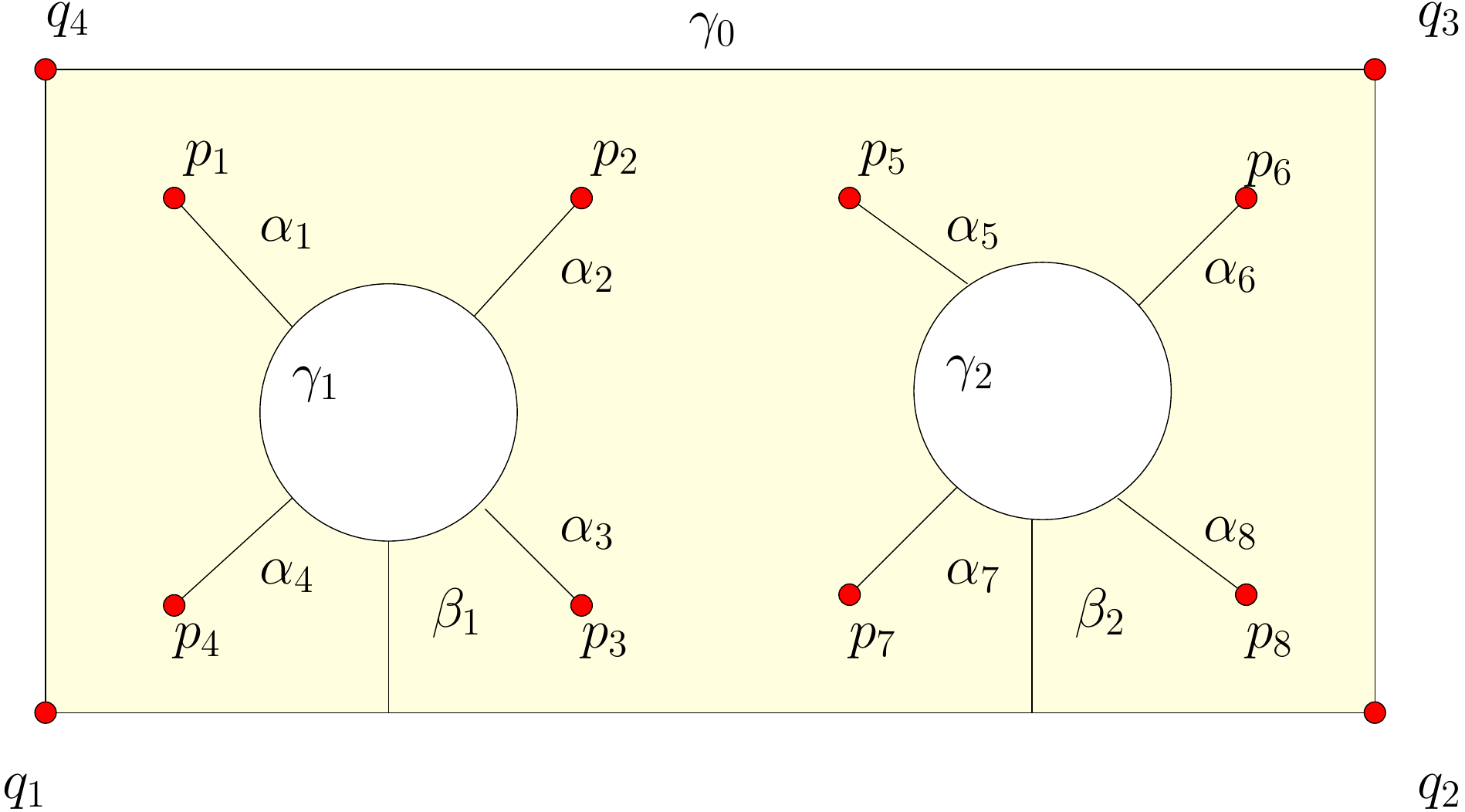}
\caption{Schematic layout.}
\label{fig:schemetic_layout}
\end{figure}

\noindent{\textbf{1. Singularity Allocation and Cut Graph}} Fig.~\ref{fig:schemetic_layout} shows a schematic 2D layout of a planar domain $\Omega$, a rectangle with two circular holes. The domain has three boundary components, $\{\gamma_0,\gamma_1,\gamma_2\}$, where $\gamma_0$ is the exterior boundary component, $\gamma_1,\gamma_2$ are inner boundary components.

\begin{figure}[h!]
\centering
\includegraphics[width=0.85\textwidth]{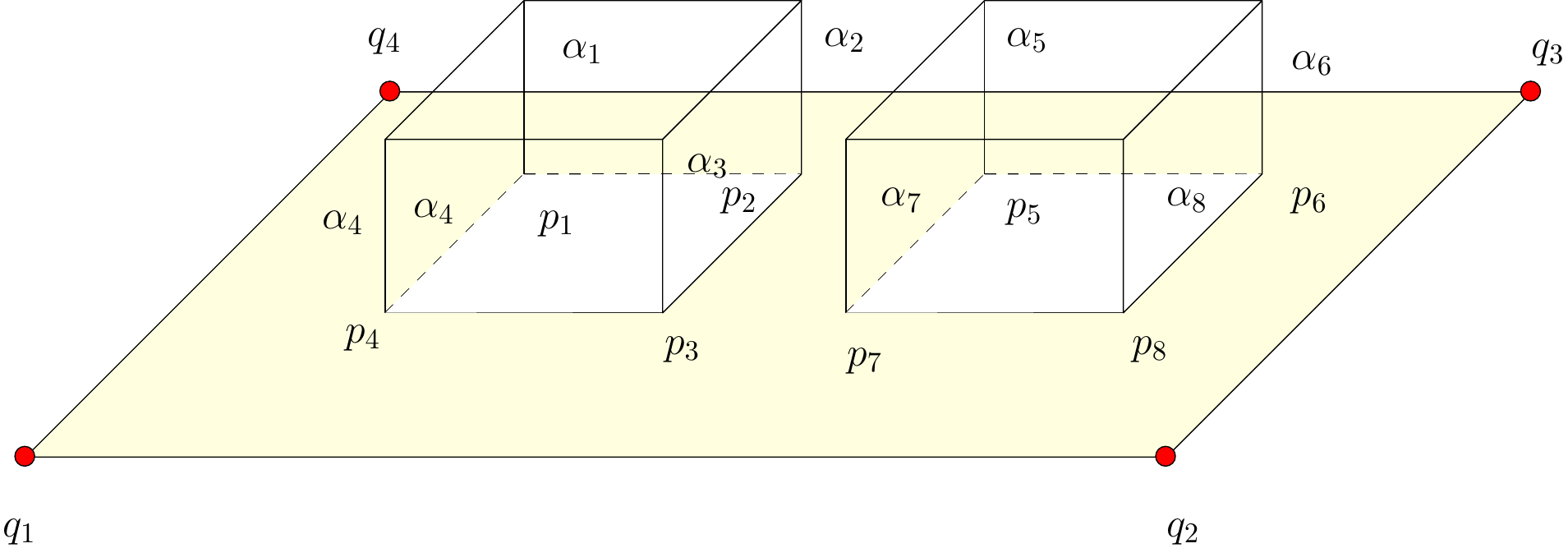}
\caption{Desired flat metric with cone singularities.}
\label{fig:embedding}
\end{figure}

Fig.~\ref{fig:embedding} shows the desired target metric with cone singularities. Following the design, four interior singularities are manually assigned, $\{p_1,p_2,p_3,p_4\}$, the valences of them are $5$. Four boundary singularities (corner points) are prescribed, $\{q_1,q_2,q_3,q_4\}$, their valences are $1$.

Fig.~\ref{fig:schemetic_layout} also shows the cut graph $L$. From each interior singularity $p_k$ to the inner boundary, we compute a shortest path $\alpha_k$. From each inner boundary component $\gamma_k$ to the exterior boundary $\gamma_0$, we draw a shortest path $\beta_k$.

\begin{figure}[h!]
\centering
\includegraphics[width=0.85\textwidth]{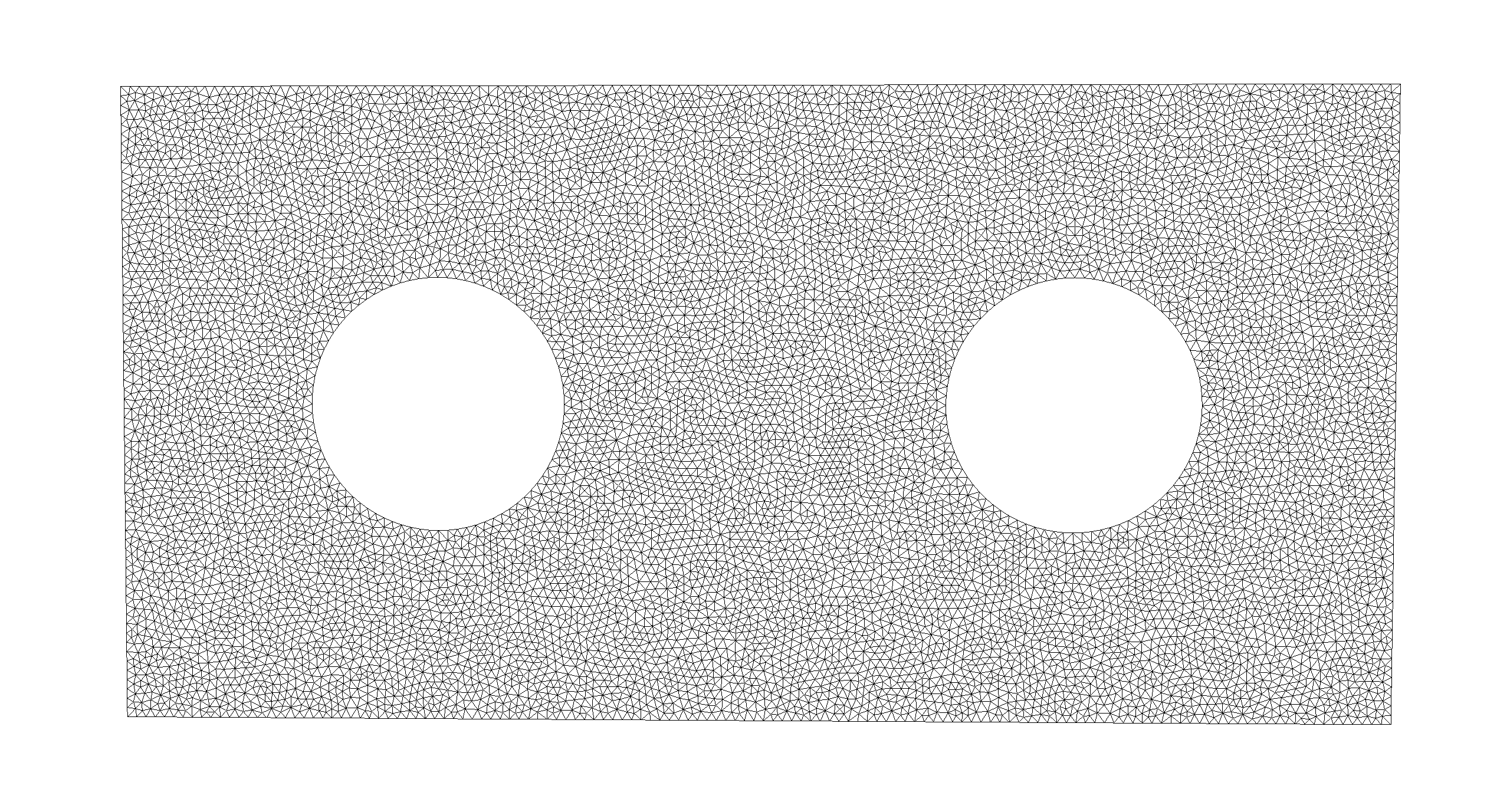}
\caption{Triangular mesh.}
\label{fig:triangulation}
\end{figure}

\noindent{\textbf{2. Ricci flow for Flat Metric With Cone Singularities}}
Fig.~\ref{fig:triangulation} shows the triangular mesh of the domain, denoted as $M$. We use conventional planar mesh generation method to triangulate the planar domain $\Omega$  using gmsh, all the inner singularities are constrained to be vertices of the triangular mesh.

We use discrete surface Ricci flow algorithm to compute a flat metric with cone singularities at the singularities, where
\[
    K(p_i) = -\frac{\pi}{2}, K(q_j) = \frac{\pi}{2}.
\]
The flat metric is denoted as $\mathbf{g}$.

\begin{figure}[h!]
\centering
\includegraphics[width=0.85\textwidth]{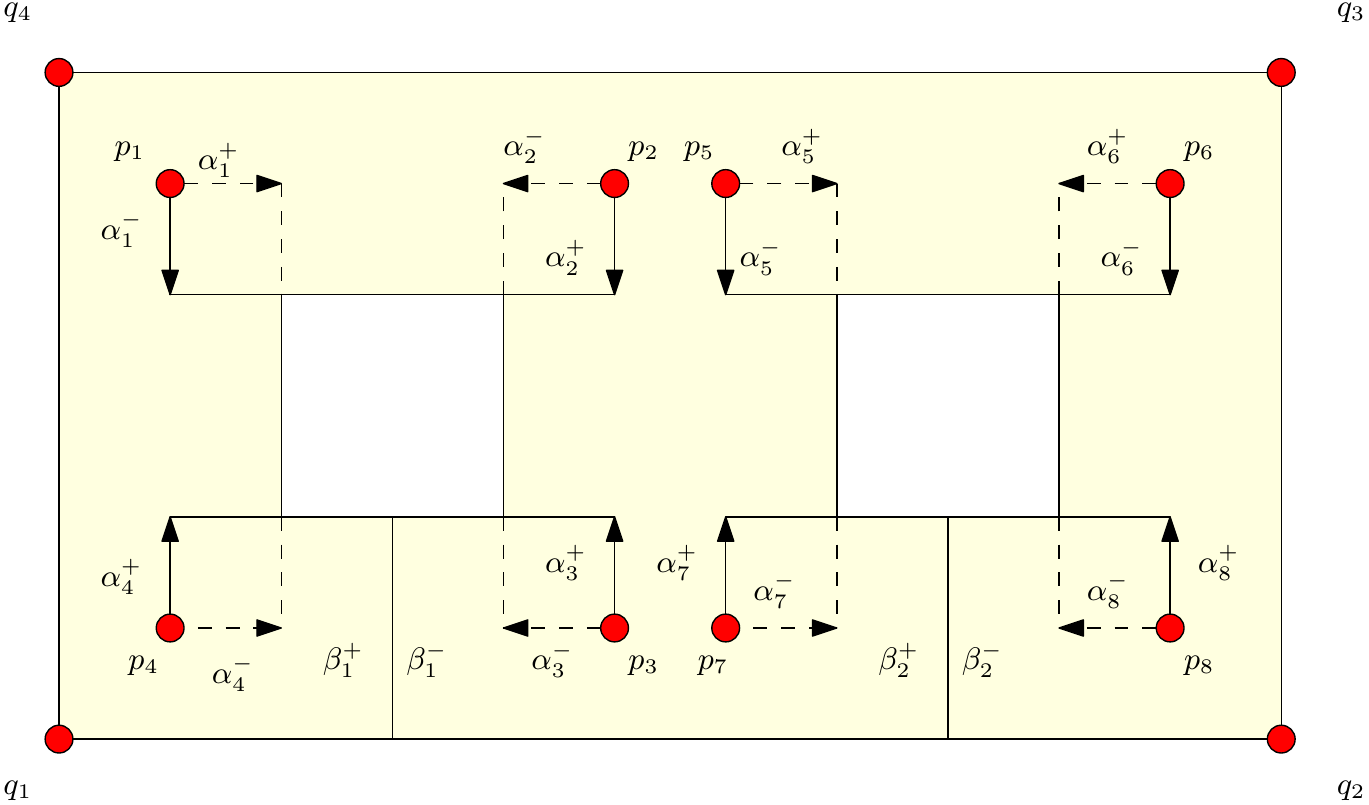}
\caption{Isometric immersion.}
\label{fig:immsersion}
\end{figure}

\noindent{\textbf{3. Isometric Immersion and Deformation}}
The surface $(\Omega,\mathbf{g})$ can be isometrically embedded in $\mathbb{E}^3$ as shown in Fig.~\ref{fig:embedding}. We slice $(\Omega,\mathbf{g})$ along the cut graph $L=\bigcup_i \alpha_i \bigcup_j \beta_j$ to obtain $\tilde{\Omega}$, which is a topological disk. Each $\alpha_k$ is split into two boundary segments $\alpha_k^+$ and $\alpha_k^-$, similarly each $\beta_k$ is split into $\beta_k^+$ and $\beta_k^-$. Then we isometrically flatten $\tilde{\Omega}$ to obtain an isometric immersion as shown in Fig.~\ref{fig:immsersion}, the immersion is denoted as $\varphi:\tilde{\Omega}\to \mathbb{R}^2$.

In this example, the singularities are carefully chosen, so that the metric $\mathbf{g}$ satisfies the all conditions in theorem \ref{thm:quad_mesh_metric}, $\alpha_k^+$ and $\alpha_k^-$ differ by a rotation of $\pi/2$, $\beta_k^+$ and $\beta_k^-$ differ by translations. Hence, the conformal structure deformation is skipped.

\begin{figure}[h!]
\centering
\includegraphics[width=0.85\textwidth]{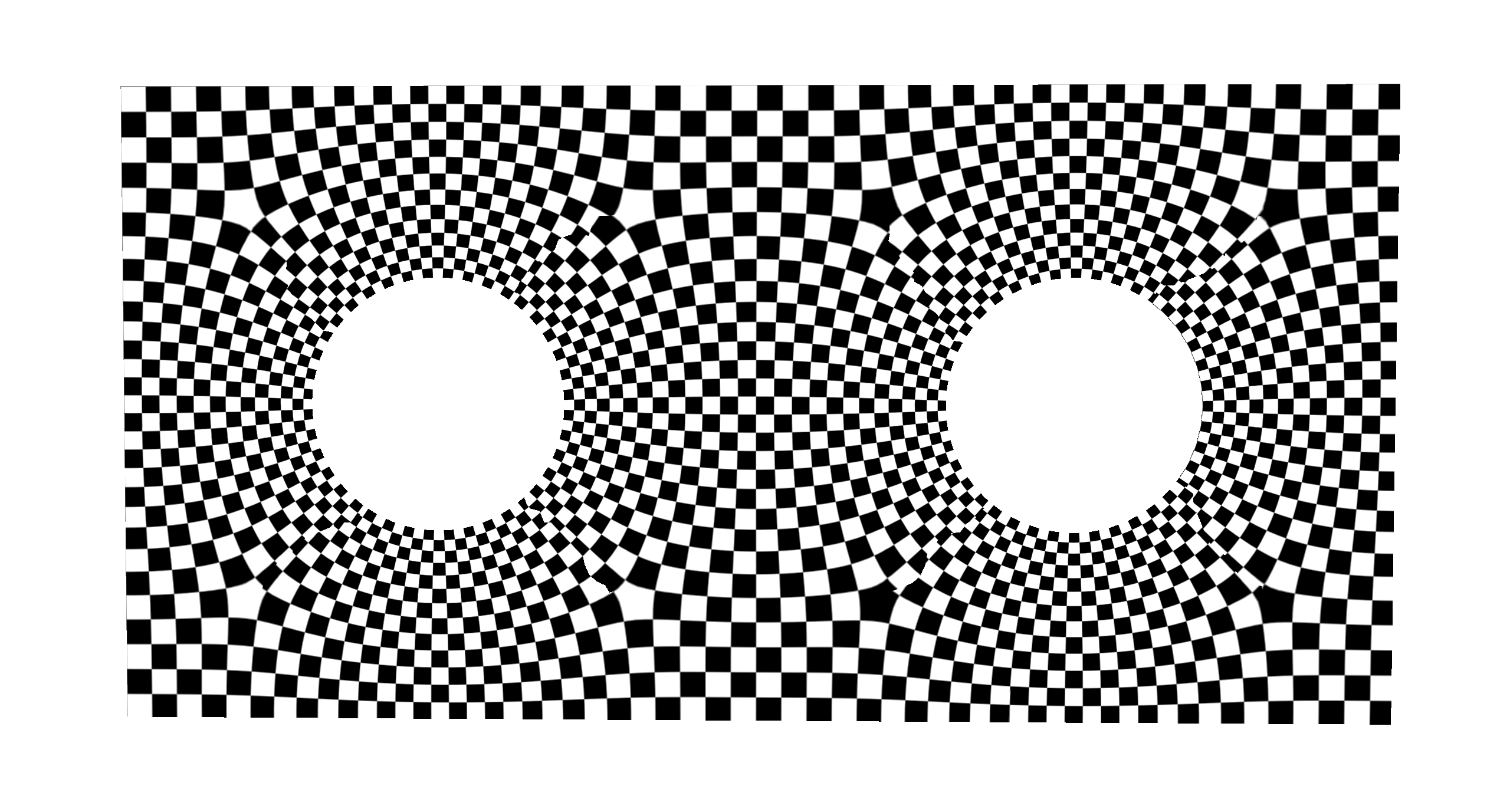}
\caption{Holomorphic quartic differential visualized by Checker board texture mapping.}
\label{fig:checker}
\end{figure}

We use the immsersion $\varphi$ as the texture coordinates for texture mapping. Fig.~\ref{fig:checker}. On the texture mapping image, we can see the horizontal and vertical geodesics.

\begin{figure}[h!]
\centering
\includegraphics[width=0.85\textwidth]{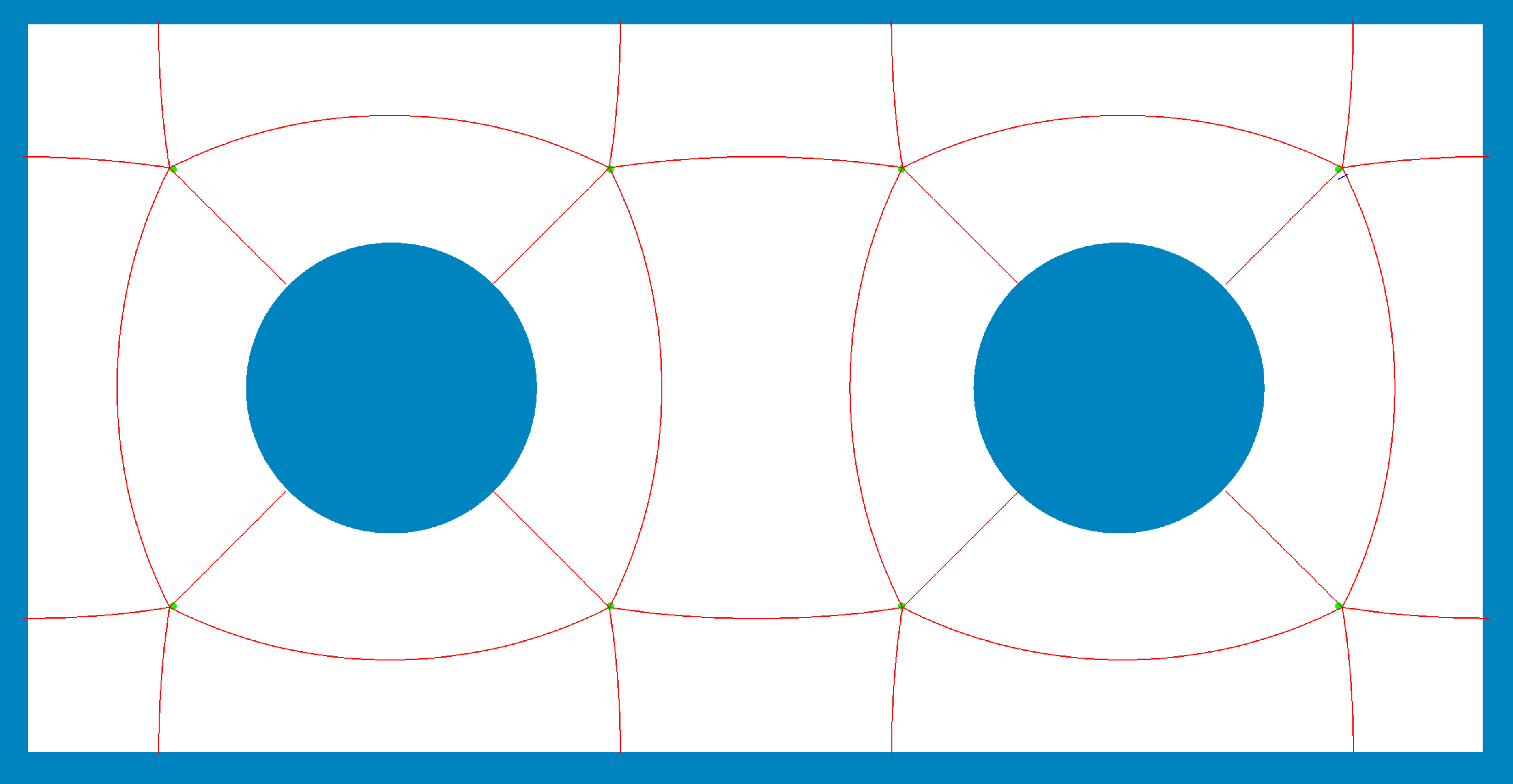}
\caption{Critical geodesics.}
\label{fig:critical_geodesics}
\end{figure}

\noindent{\textbf{3. Geodesics and Quad-meshing}}

The critical geodesics are traced, which either connect different singularities, or orthogonal to the boundaries. The critical geodesics give the skeleton of the quad-mesh as shown in Fig.~\ref{fig:critical_geodesics}.

\begin{figure}[h!]
\centering
\includegraphics[width=0.85\textwidth]{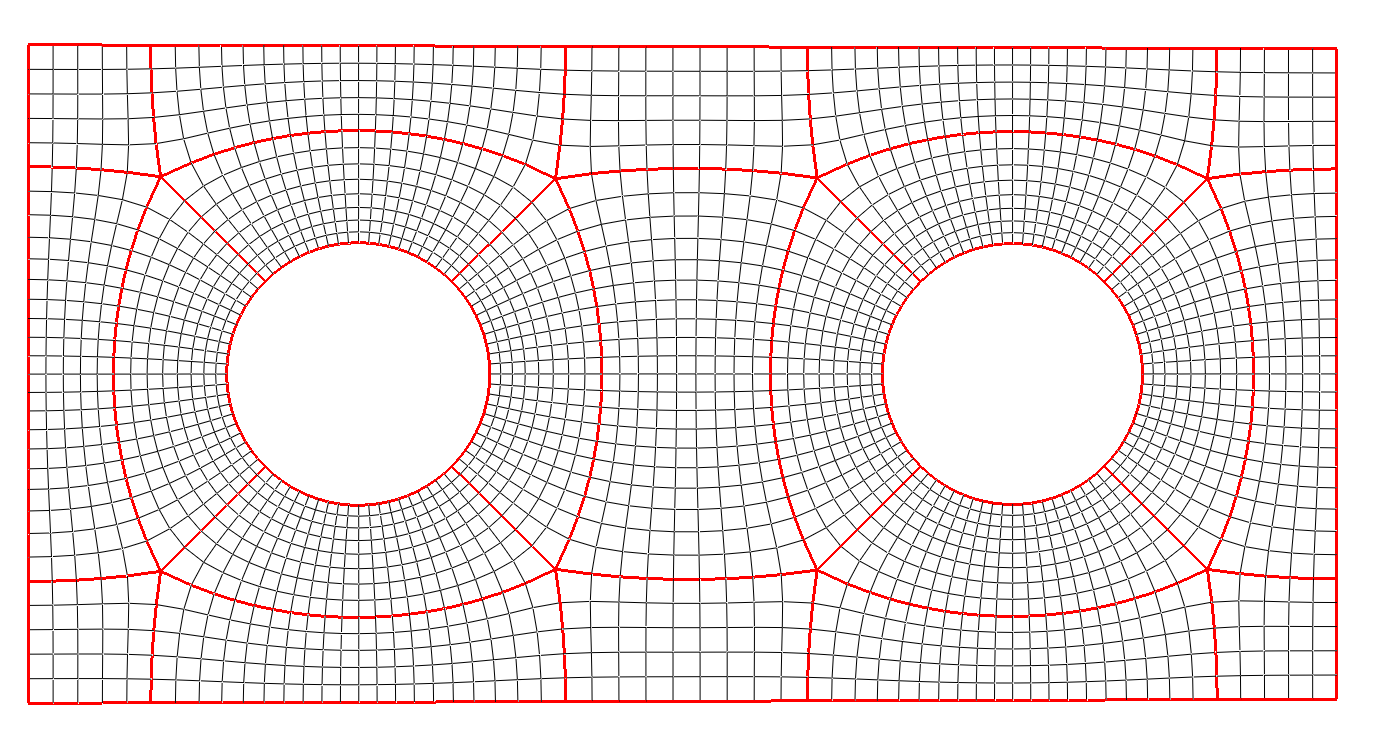}
\caption{The quad-mesh induced by geodesics.}
\label{fig:quad_mesh}
\end{figure}

We then refine the skeleton by tracing more geodesics on $(\Omega,\mathbf{g})$, which are orthogonal or parallel to the critical geodesics. The geodesics induce the quadrilateral meshing as shown in Fig.~\ref{fig:quad_mesh}. 
\section{Experimental Results}
\label{sec:experiments}

All the experiments were conducted on a PC with 1.90GHz Intel(R) core(TM) i7-8650U CPU and 64-bit Windows 10 operating system. The running time is reported in table \ref{tab:time}.

\begin{figure}[h!]
\centering
\begin{tabular}{c}
\includegraphics[width=0.75\textwidth]{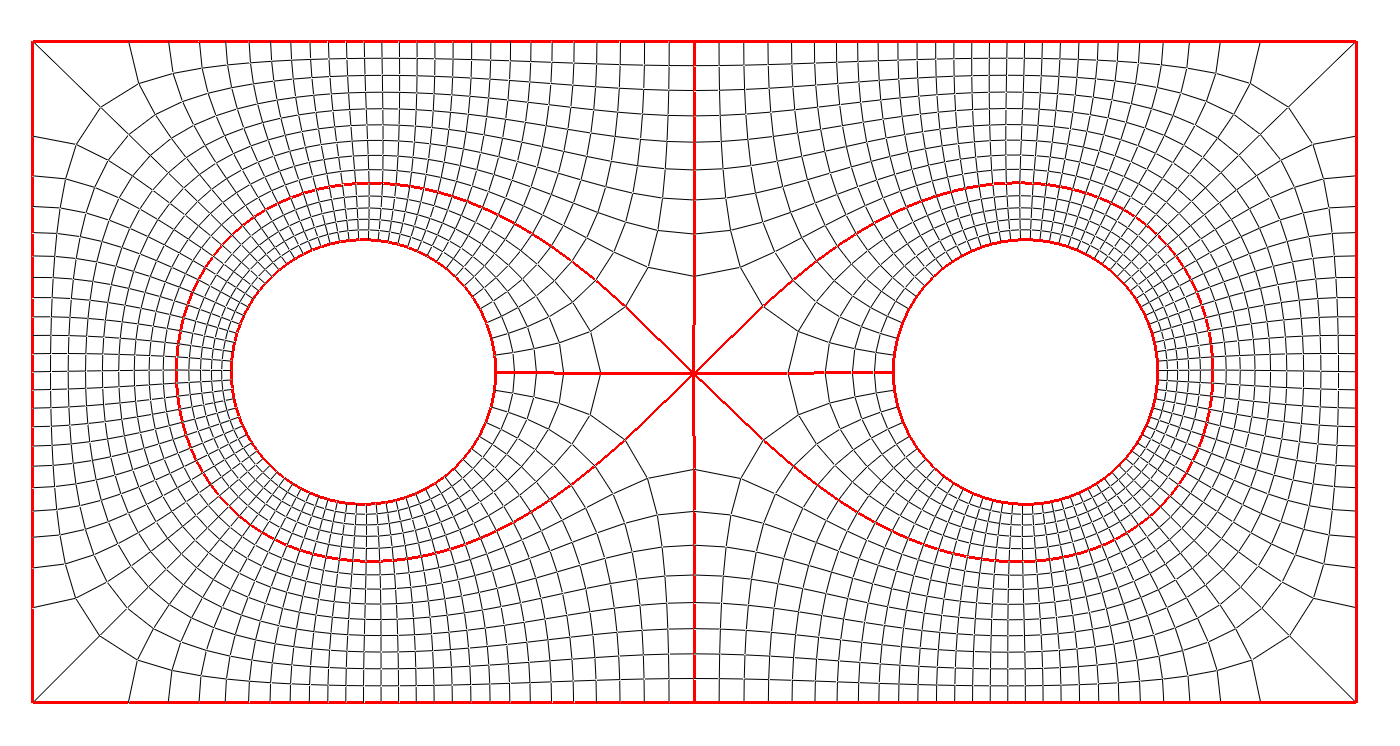}\\
\includegraphics[width=0.75\textwidth]{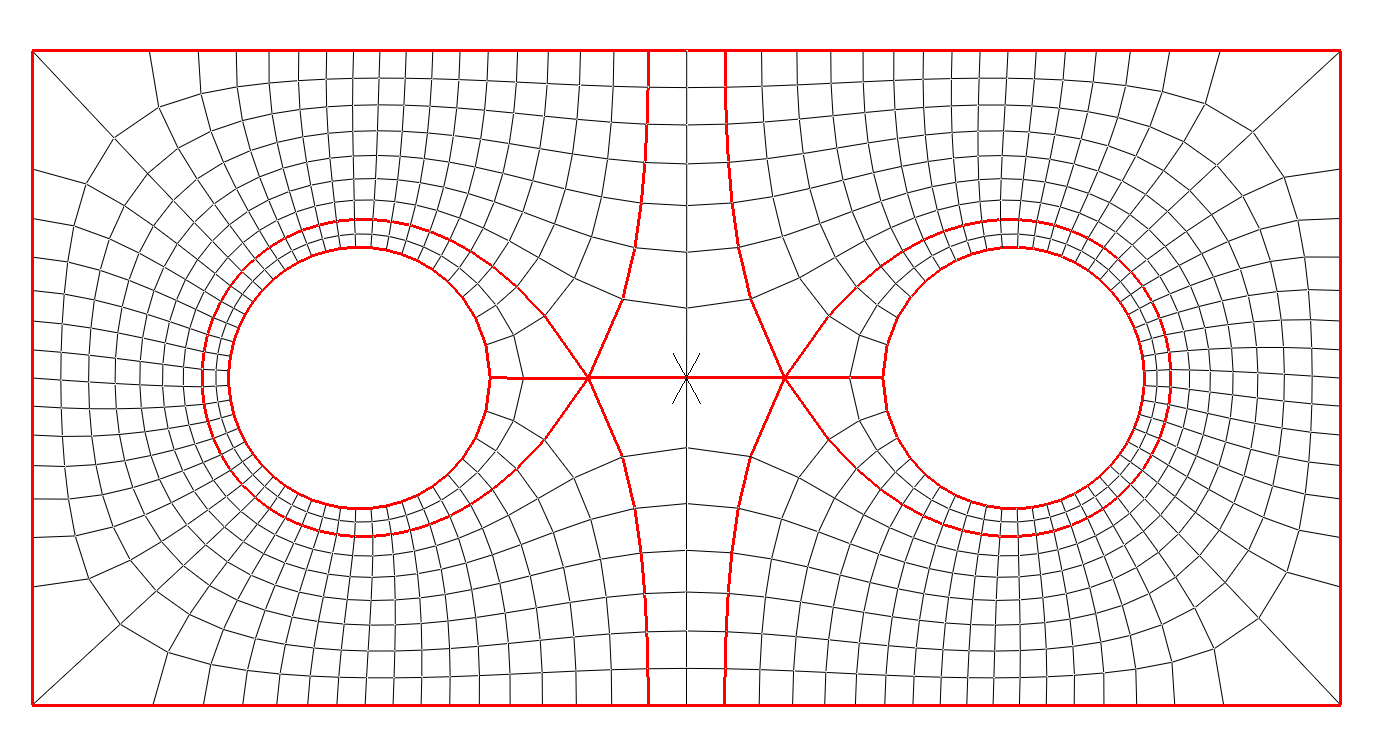}\\
\includegraphics[width=0.75\textwidth]{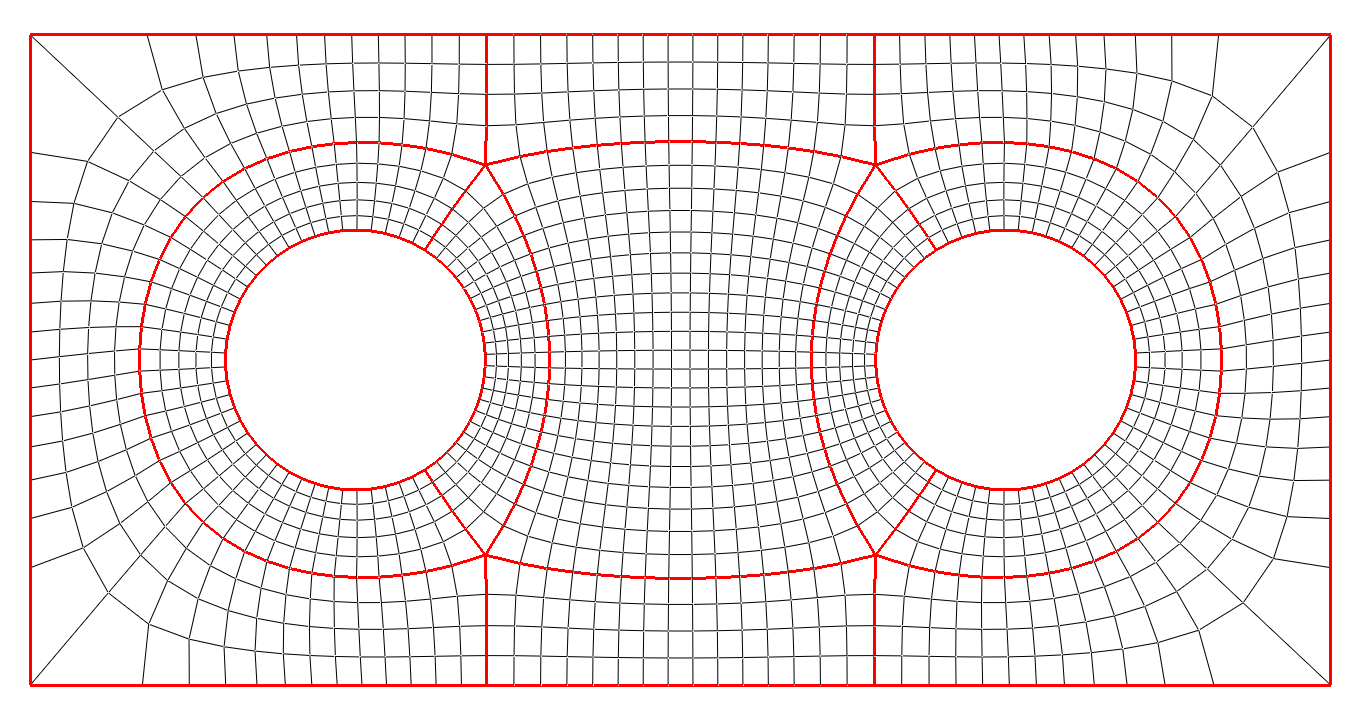}
\end{tabular}
\caption{A quad-mesh with 1,2 and 4 singularities on a planar rectangle with 2 circular inner holes.}
\label{fig:rectangle}
\end{figure}

Fig.~\ref{fig:rectangle} illustrates the quad-meshes obtained by the proposed method with different number of singularities of the planar rectangle with two circular inner holes.

\begin{figure}[h!]
\centering
\begin{tabular}{cc}
\includegraphics[width=0.5\textwidth]{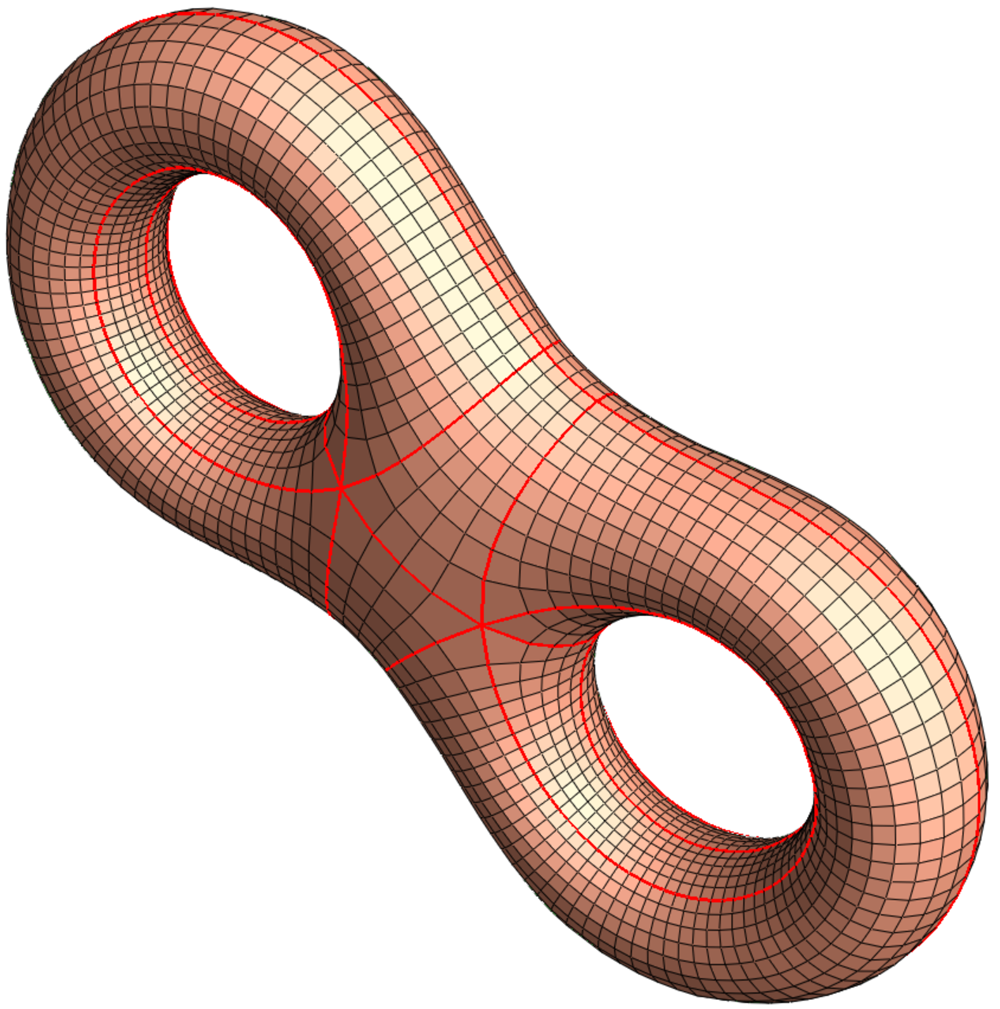}&
\includegraphics[width=0.5\textwidth]{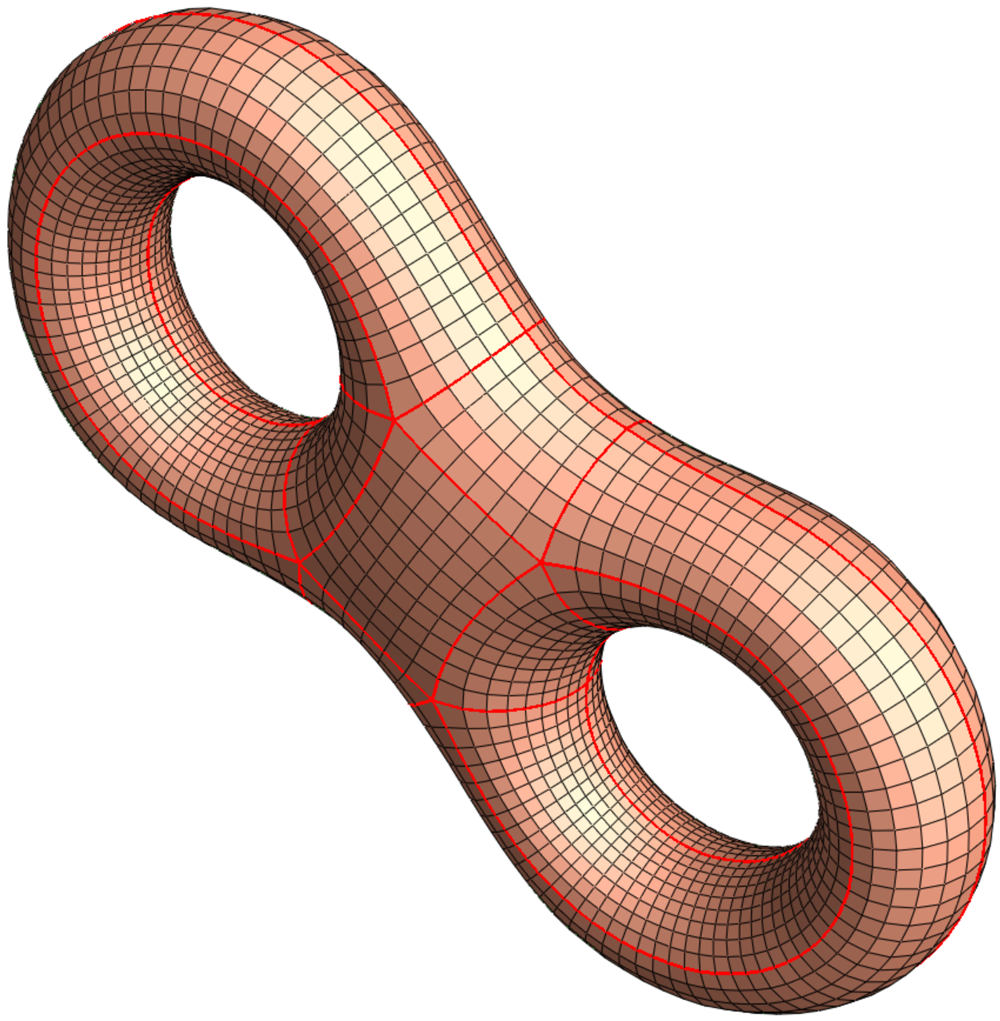}
\end{tabular}
\caption{A quad-mesh with 4 and 8 singularities on a genus 2 surface.}
\label{fig:eight2}
\end{figure}

The proposed method can be applied for surfaces with arbitrary typologies. Fig.~\ref{fig:eight1} and Fig.~\ref{fig:eight2} show the quad-meshes with different number of singularities of a genus two closed surface.

\begin{table}[]
    \centering
    \begin{tabular}{|c|c|c|c|}
    \hline
    Model & \# Vertices & \# Singularities & time (ms) \\
    \hline
    \hline
    Rectangle with holes &  10529 & 1 & 30936  \\
    \hline
    Rectangle with holes & 10675 & 2 & 40753 \\
    \hline
    Rectangle with holes & 10641 & 4 & 24510 \\
    \hline
    Rectangle with holes & 10668 & 12 & 18733\\
    \hline
    eight mesh& 16061 & 4 & 14988\\
    \hline
    eight mesh& 16060 & 8 & 12427\\
    \hline
    \end{tabular}
    \caption{Computational time.}
    \label{tab:time}
\end{table}

The computation of triangular mesh generation, Ricci flow, tracing critical trajectories (and generating skeletons) are fast enough to allow the users to modify the singularity positions interactively. The geodesic tracing is accurate the stable, so the computation of skeletons is straightforward without any manual refinement. Therefore, the whole computational pipeline is automatic expect the initial step for singularity location. 

\section{Conclusion}
\label{sec:conclusion}

%``I always thought something was fundamentally wrong with the universe'' \citep{adams1995hitchhiker}

This work proposes a novel metric based algorithm for quadrilateral mesh generating. Each quad-mesh induces a Riemannian metric satisfying special conditions: the metric is a flat metric with cone signualrites conformal to the original metric, the total curvature satisfies the Gauss-Bonnet condition, the holonomy group is a subgroup of the rotation group $\{e^{ik\pi/2}\}$, furthermore there is cross field obtained by parallel translation which is aligned with the boundaries, and its streamlines are finite geodesics. Inversely, such kind of metric induces a quad-mesh. Based on discrete Ricci flow and conformal structure deformation, one can obtain a metric satisfying all the conditions and obtain the desired quad-mesh.

This method is rigorous, simple and automatic. Our experimental results demonstrate the efficiency and efficacy of the algorithm.

In the future, we will explore the methods to automatically determined the positions and indices of all the singularities based on Abel differentials.

%\bibliographystyle{plain}
%\bibliography{references,quad}
\end{document}